\PassOptionsToPackage{dvipsnames}{xcolor}
\documentclass[journal,twoside,web]{ieeecolor}
\usepackage{generic}

\def\BibTeX{{\rm B\kern-.05em{\sc i\kern-.025em b}\kern-.08em
	T\kern-.1667em\lower.7ex\hbox{E}\kern-.125emX}}
\let\labelindent\relax

\usepackage[inline]{enumitem}
\usepackage[hidelinks]{hyperref}

\usepackage{etoolbox}
\newtoggle{paperSC}
\newcommand{\condVSpace}[1]{\iftoggle{paperSC}{}{\vspace{#1}}}

\usepackage[dvipsnames]{xcolor}
\usepackage{tikz}
\usepackage{tikz-3dplot}
\usetikzlibrary{angles}

\usepackage{graphicx}
\usepackage{subcaption}

\newcommand{\diff}[1]{#1}

\usepackage[noadjust]{cite}

\usepackage{mathtools}
\usepackage{amssymb}
\usepackage{amsthm}

\theoremstyle{plain}
\newtheorem{theorem}{Theorem}
\newtheorem{proposition}{Proposition}
\newtheorem{corollary}{Corollary}[theorem]
\newtheorem{lemma}{Lemma}
\newtheorem{assumption}{Assumption}

\newtheorem{definition}{Definition}

\theoremstyle{definition}

\def\qedremark{\hspace*{\fill} $\triangleleft$\par\endtrivlist\unskip}
\newtheorem{remarkUnwrapped}{Remark}
	{\begin{remarkUnwrapped}}%
	{\qedremark\end{remarkUnwrapped}}

\def\qedexample{\hspace*{\fill} $\triangleleft$\par\endtrivlist\unskip}
\newtheorem{exampleUnwrapped}{Example}
\newenvironment{example}%
	{\begin{exampleUnwrapped}}%
	{\qedexample\end{exampleUnwrapped}}

\newcommand{\N}{\mathbb{N}}
\newcommand{\Z}{\mathbb{Z}}

\newcommand{\R}{\mathbb{R}}
\newcommand{\C}{\mathbb{C}}

\newcommand{\T}{\mathbb{T}}

\DeclarePairedDelimiter\pr{\lparen}{\rparen}
\DeclarePairedDelimiter\ps{\lbrack}{\rbrack}
\DeclarePairedDelimiter\pc{\lbrace}{\rbrace}

\DeclarePairedDelimiterX\norm[1]{\lvert}{\rvert}{
	\ifblank{#1}{\,\cdot\,}{#1}
}

\DeclarePairedDelimiterX\inner[2]{\langle}{\rangle}{
	\ifblank{#1}{\,\cdot\,}{#1},
	\ifblank{#2}{\,\cdot\,}{#2}
}

\DeclarePairedDelimiterX\RVvar[1]{.}{.}{
	 #1
}

\DeclarePairedDelimiterXPP\RVpdf[2]{#1}{\lparen}{\rparen}{}{
	 #2
}

\DeclarePairedDelimiterXPP\RVexpected[1]{\mathbb{E}}{\lbrack}{\rbrack}{}{
	 #1
}

\DeclarePairedDelimiterXPP\RVprob[1]{\mathbb{P}}{\lbrace}{\rbrace}{}{
	 #1
}

\newcommand{\st}{\,:\,}
\DeclarePairedDelimiterX\set[1]{\lbrace}{\rbrace}{
	\renewcommand\st{\, \delimsize\vert \, \mathopen{}}
	#1
}

\usepackage{silence}
\settoggle{paperSC}{true}

\begin{document}

\title{Frequency Response of Nonlinear Systems: Notions, Analysis, and Graphical Representation}

\author{
	Alessio Moreschini
	and
	Matteo Scandella
	\thanks{%
	Alessio Moreschini is with the Department of Electrical and Electronic Engineering, Imperial College London, SW72AZ London, U.K. (e-mails: \texttt{a.moreschini@imperial.ac.uk}).
	Matteo Scandella is with the Department of Management, Information and Production Engineering, University of Bergamo, via Marconi 5, 24044, Dalmine (BG), Italy (email: \texttt{matteo.scandella@unibg.it}).
	}
}

\maketitle

\begin{abstract}
The invariance principle, through which the steady-state behavior of nonlinear systems was introduced by Isidori and Byrnes, is leveraged in this article to bring forth a unifying characterization of the frequency response of nonlinear systems.
We show that, for systems under nonlinear periodic excitations, the frequency response can still be defined as a complex-valued function in a phasor form.
However, together with suitable notions of gain and phase functions, we show the existence of another function that completes the frequency response and allows quantifying the distortion introduced by the system in the steady-state output.
This nonlinear characterization enabled the representation over input frequency and amplitude of the gain, phase, and distortion produced by the system, via a nonlinear enhancement of the Bode diagrams.
This graphical representation of the frequency response is well-suited to performance analysis of a nonlinear system and, furthermore, allows for the formulation of the loop-shaping problem for nonlinear systems.

\end{abstract}

\section{Introduction}
\label{sec:intro}

\iftoggle{paperSC}
{Frequency}
{\IEEEPARstart{F}{requency}}
response methods for analysis and synthesis of control systems carry a century-long tradition~\cite{MacFarlane1979perspectives}, and they remain prominent in numerous engineering applications~\cite{baillieul2021encyclopedia}, including power systems~\cite{huang2024gain,xu2025interharmonic,gibbard2015small,gu2022power}, systems biology~\cite{olsman2019hard,khammash2022cybergenetics,del2015biomolecular,keener2025mathematical}, modern circuit theory~\cite{paul2007analysis,bullmore2009complex,selvaratnam2025frequency,chua2024memristors}, signal processing~\cite{ortega2018graph,boashash2015time,dong2019learning}.

At the very heart of the frequency response of linear time-invariant (LTI) systems lies the harmonic-preserving principle, whereby a system subjected to harmonic excitation produces a harmonic output at the same frequency, differing only in amplitude and phase.
This principle enabled quantitative measures of how the system response rolls-off or rolls-up the input (\emph{amplitude response}) and how phase leads or lags (\emph{phase response}) across frequencies~\cite{skogestad2005multivariable,franklin2025feedback,doyle2013feedback}.
These measures have had a \diff{profound influence} on analyzing feedback systems through graphical representations such as Bode~\cite{bode1945network} and Nyquist~\cite{Nyquist1932regeneration} diagrams, and on control synthesis, which enables the design of feedback systems with guaranteed performance directly from the shape of the frequency response.
For instance, in loop-shaping controller synthesis, closed-loop specifications can be enforced by feedback control whose Bode magnitude and phase shape the loop frequency response function, see, \emph{e.g.},~\cite{doyle2013feedback,zames2003feedback,tannenbaum1980feedback,mcFarlane1992loop}.

Although frequency response and its graphical representations are widespread in LTI systems, their enhancement to nonlinear control systems remains vague~\cite{bechhoefer2011kramers}.
The crux of nonlinear dynamics is that the output response relies on both the input frequency and amplitude, producing complex effects such as harmonic distortion and frequency-dependent nonlinearities.
Due to the limitations of the Laplace transform in converting nonlinear systems into the frequency domain, a unifying characterization of the frequency response of nonlinear systems remains largely unexplored.
However, before exploring the \emph{status quaestionis} for nonlinear systems, we illustrate how the frequency response of an LTI system can be recast in the time domain. %

\condVSpace{0.5em}

\subsubsection*{Frequency response by an invariance principle}

Consider an LTI system governed by the equations
\begin{align} \label{eq:LTI-system}
	\dot{x} & = Ax + Bu, &
	y & = Cx,
\end{align}
with state $x(t)\in\R^n$, input $u(t)\in\R$, output $y(t)\in\R$, and matrices $A$, $B$, and $C$ of appropriate dimensions, defining a minimal realization.
Suppose that $A$ is Hurwitz, \emph{i.e.}, $\sigma(A)\subset\C^-$.
Applying the Laplace transform to~\eqref{eq:LTI-system} yields the complex-valued (transfer) function $H(s)\coloneq C(sI-A)^{-1}B$, and the \emph{frequency response function} $H(j\varpi) \coloneq H(s)|_{s=j\varpi}$, with frequency $\varpi\in\R_+$. When system~\eqref{eq:LTI-system} is driven by a sine wave of the form $u(t) = a_u \sin(\varpi t + \phi_u)$, with amplitude $a_u\in\R$, and phase $\phi_u\in\R$, the output response is a sine wave of the form $y(t) = a_y\sin(\varpi t + \phi_y)$, with amplification factor $a_y\in\R$ and phase-shifted factor $\phi_y\in\R$ given by, see~\cite[Sec. 2.1]{skogestad2005multivariable},
\begin{align*}
	a_y &= a_u \norm*{H(j\varpi)}, &
	\phi_y &= \phi_u + \arg \pr*{ H(j\varpi) }.
\end{align*}

The same frequency response function can be alternatively recast in the time domain by leveraging the notion of \emph{steady-state} \diff{Ã  la} Isidori and Byrnes~\cite{isidori1990output,isidori1995nonlinear,isidori2008steady}, in which the output response is analyzed through an invariance principle associated with the cascade interconnection of the system and a harmonic oscillator describing the sine wave, see~\cite[Sec. A.5]{isidori2017lectures}.
Specifically, sine waves of the form $u(t) = a_u\sin(\varpi t + \phi_u)$, can be alternatively generated by a harmonic oscillator of the form
\begin{align} \label{eq:LTI-exosystem}
	\dot{z} & = S(\varpi)z, & u & = Lz,
\end{align}
with initial conditions $z(0) = (a_u\sin(\phi_u), a_u \cos(\phi_u))\in\R^2$ and matrices
\begin{align} \label{eq:fixed_SL}
	S(\varpi)
	&
	= \ps*{ \begin{array}{cc} 0 & \varpi \\ -\varpi & 0 \end{array} },
	&
	L
	&
	= \ps*{ \begin{array}{cc} 1 & 0 \end{array} }.
\end{align}

Since $A$ is Hurwitz and $S(\varpi)$ has purely imaginary eigenvalues for every $\varpi>0$, it is not hard to establish that, irrespective of the initial condition $x(0)$, the response of the system~\eqref{eq:LTI-system}, driven by the harmonic oscillator~\eqref{eq:LTI-exosystem},
converges to the function $Y(\varpi) \coloneq C \Phi(\varpi) $, where $\Phi(\varpi)$ is the unique solution of the Sylvester equation
\begin{equation} \label{eq:Sylvester}
	A\Phi(\varpi) + BL = \Phi(\varpi) S(\varpi).
\end{equation}
The function $Y(\varpi)$, which is determined by the parameterized center eigenspace of the cascade system~\eqref{eq:LTI-system}-\eqref{eq:LTI-exosystem}, is directly related to the frequency response function $H(j\varpi)$, since $Y(\varpi) \iota = H(j\varpi)$ for $\iota \coloneq [1, j]^\top$, see Lemma~\ref{th:LTI-tf-Phi} in Appendix.

With this in mind, the foundational question that motivated this article is whether this time-domain characterization of the frequency response can be extended to nonlinear systems, enabling frequency-based tools for analysis and design analogous to those used in linear systems theory.

\condVSpace{0.5em}

\subsubsection*{State of the Art}
\label{sec:intro:literature}

Significant efforts have been dedicated over the last three decades to the nonlinear enhancements of frequency-domain concepts.
One of the earliest extensions of the frequency response to nonlinear systems is provided through \emph{generalized frequency response functions}.
For convergent Volterra series representations of its input-output map describing the system, these functions can be deduced from the Fourier transforms of the Volterra kernels.
Although they extend the notion of frequency response, the presence of frequency coupling and harmonic interactions makes their interpretation considerably more complex than the linear case, see, \emph{e.g.},~\cite{george1959continuous,bayma2018analysis,boyd2003fading}.
These functions are challenging to compute and represent, limiting their practical applicability.

From a control-theoretic perspective, the attack on the problem of $H_\infty$-control of nonlinear systems, as presented in~\cite{isidori1992disturbance}, produced a tool that, based on the steady-state response of a forced system, can be recognized as an early characterization of the frequency response of nonlinear systems.
In~\cite{isidori1992disturbance}, the gain (or amplitude) of a nonlinear system has been given by the ratio of the $\mathcal{L}_2$-induced norm of the steady-state output and that of the periodic input, providing an instrumental tool for extending $H_\infty$-control methods to nonlinear systems.
In the context of moment matching,~\cite{astolfi2010model} characterized the frequency response of a nonlinear system under harmonic excitations as the ratio of the Fourier series (with nonnegative indices) of the steady-state output to that of the output of the harmonic oscillator.
As shown in~\cite{pavlov2007frequency}, convergent systems driven by harmonic inputs admit a global steady-state behavior, at various input amplitudes and frequencies.
This property led to the definition of gain as the maximum output-to-input amplitude ratio over a finite set of harmonic input amplitudes.
The gain in~\cite{pavlov2007frequency} was illustrated via a gain-frequency diagram, offering a Bode-like representation of how a convergent system amplifies sinusoidal inputs across frequencies and amplitudes.
Despite significant interest and advancements in gain for nonlinear systems, phase received much less attention in these studies.

The first interest in the phase of nonlinear systems was associated with frequency-domain approximations~\cite{chua1979frequency,billings1994analysing,rugh1981nonlinear}, and operator theory~\cite{zames1966inputII}.
Recently, there has been a resurgence of interest in the notion of phase for nonlinear systems from an operator perspective.
In~\cite{chen2021phasenonlinearsystems,chen2021singularanglenonlinearsystems}, the phase is defined as the angle between the output of the system and the complex signal constructed using the input and its Hilbert transform as its real and imaginary part, respectively.
The Hilbert transform has been the subject of many research areas, such as optics, fluids, and signal processing~\cite{king2009hilbert}.
A similar operator reasoning is employed in~\cite{chaffey2023graphical}, where the concepts of gain and phase are defined using a scaled relative graph~\cite{ryu2022scaled}, which opened to a nonlinear operator generalization of the Nyquist diagram, shedding new light on how the system amplifies the input and how it shifts the signal.
However, operator-based representations can still be less practical than first-order models, as capturing complex nonlinear behavior in operator form is often more challenging.

\condVSpace{0.5em}

\subsubsection*{Contributions}
\label{sec:intro:contribution}

Taking as a point of departure the invariance principle, a characterization of the frequency response function for nonlinear systems is given.
For a family of periodic nonlinear inputs parameterized by $\omega$, it is shown that the \emph{frequency response function} associated with the nonlinear system can be given as the complex-valued function
\begin{equation*}
	\boxed{ \Gamma(\omega) = \alpha(\omega)r(\omega) e^{j \vartheta(\omega)} , }
\end{equation*}
where $\alpha(\omega)$, $\vartheta(\omega)$, and $r(\omega)$, which are referred to as the $\omega$-gain, $\omega$-phase, and $\omega$-radius, respectively, are functions associated with the steady-state relation between input and output.
As a result, it is shown that most of the features of the frequency response of an LTI system are expressions of more general nonlinear principles.
Indeed, the introduced notion captures amplitude-dependent nonlinear effects in both amplitude and phase by reflecting attenuation or amplification of the input effect ($\omega$-gain), leads or lags of the input effect ($\omega$-phase), and nonlinear distortion of the input signal through the system ($\omega$-radius).
While $\omega$-gain and $\omega$-phase are enhancements of the notions of gain and phase of LTI systems, the $\omega$-radius is a function, purely resulting from nonlinear dynamics, which quantifies the distortion introduced by the system in the steady-state output compared to the periodic input.

Together, these frequency response tools allow for graphical representation of input frequency and amplitude.
This provides a full picture of the frequency response function in the form of a nonlinear Bode diagram.
It also facilitates graphical analysis of nonlinear system dynamics and bridges the gap between linear approximations and true nonlinear behavior.
Finally, we propose a \emph{nonlinear loop-shaping} formulation in terms of a specification set which accounts for desired quantitative requirements on the feedback system.
The preliminary problem of shaping only the gain function was initialized in the conference article~\cite{moreschini2026onfrequency}.

\condVSpace{0.5em}

\subsubsection*{Organization}
\label{sec:intro:organization}
Section~\ref{sec:w-response} opens with the introduction of basic definitions and the setup of the framework.
Section~\ref{sec:nonlinear-freq-analysis} formalizes the notion of frequency response, along with the functions $\omega$-gain, $\omega$-phase, and $\omega$-radius, before illustrating nonlinear Bode diagrams through examples.
Section~\ref{sec:analysis_of_frf} analyzes the frequency response of dissipative systems, before presenting a weak form of the superposition principle.
Section~\ref{sec: loop-shaping} discusses the nonlinear loop-shaping problem and validates it in an example.
Final remarks and open challenges are in Section~\ref{sec:end}.

\condVSpace{0.5em}

\subsubsection*{Notation}
\label{sec:intro:notation}

The sets of real, natural, integer, and complex numbers are denoted by $\R$, $\N$, $\Z$, and $\C$, respectively, and $j \in \C$ denotes the imaginary unit.
Given a function $\mathcal{A}$ with domain $Q$ mapping into sets, $\bigtimes_{q \in Q} \mathcal{A}(q)$ denotes the Cartesian product of its values.
The set of strictly positive real numbers is denoted by $\R_+$, and the set of $n \times m$ real matrices by $\R^{n \times m}$.
Tuples of real numbers and column vectors are used interchangeably; when clear from the context, $0$ denotes the zero matrix of appropriate dimension.
For $\R^n$, $\norm*{} : \R^n \to \R$ denotes the Euclidean norm.
The space of square-integrable functions mapping $\mathcal{A} \subseteq \R^n$ to $\mathcal{B} \subseteq \R^m$ is denoted by $\mathcal{L}_2(\mathcal{A},\mathcal{B})$, \emph{i.e.}, the set of measurable functions $f : \mathcal{A} \to \mathcal{B}$ such that $\int_{\mathcal{A}} \norm*{f(x)}^2 dx < \infty$.
When $\mathcal{B}$ is clear from the context, $\mathcal{L}_2(\mathcal{A})$ is used in place of $\mathcal{L}_2(\mathcal{A},\mathcal{B})$.
For $f,g \in \mathcal{L}_2(\mathcal{A},\mathcal{B})$, the inner product and induced norm are, respectively, defined as
\begin{align*}
	\inner*{f}{g}_{\mathcal{L}_2(\mathcal{A})}
	&
	\coloneq
	\int_{\mathcal{A}} f(x)^\top g(x) dx \, ,
	&
	\norm*{f}_{\mathcal{L}_2(\mathcal{A})}
	&
	\coloneq
	\sqrt{\inner*{f}{f}_{\mathcal{L}_2(\mathcal{A})} } \, .
\end{align*}

\section{Basic Definitions}

\label{sec:w-response}

Consider a multi-input, multi-output nonlinear system described in state-space by the equations of the form
\begin{equation} \label{eq:driven_sys}
\begin{aligned}
	\dot{x} & = f(x,u), &
	y & = h(x,u),
\end{aligned}
\end{equation}
with state $x(t) \in \R^n$, input $u(t) \in \R^m$, and output $y(t) \in \R^p$.
To guarantee that, for all $x_0 \coloneq x(0) \in \mathcal{X} \subseteq \R^n$ the solution of~\eqref{eq:driven_sys} exists on some open interval containing $t=0$ and is unique, we assume that $f : \R^n \times \R^m \to \R^n$ and $h : \R^n \times \R^m \to \R^p$ are locally Lipschitz in $x$ and uniformly continuous in $u$, and that $u(t)$ is measurable and essentially bounded over any finite interval. Frequency response analysis for linear systems focuses on the steady-state response to harmonic input signals.
Hence, characterizing a nonlinear frequency response requires a corresponding notion of steady-state suitable for nonlinear systems.
Following Isidori and Byrnes~\cite{isidori2008steady}, the steady-state response for nonlinear systems of the form~\eqref{eq:driven_sys} can be defined under the assumption that the input $u(t)$ is a periodic signal generated by a signal generator described by the equations%
\footnote{
	Since frequency and amplitude for nonlinear oscillators might be contingent on the initial conditions other than system parameters (see, \emph{e.g.},~\cite{he2006some,cveticanin2018strong}), it is appropriate to parameterize $\ell$ by $\omega$.
	This is relevant in situations in which~\eqref{eq:driving_sys} can only be initialized within its omega-limit set to induce oscillatory behaviors. %
}
\begin{equation} \label{eq:driving_sys}
\begin{aligned}
	\dot{z} & = s(\omega, z), &
	u & = \ell(\omega, z),
\end{aligned}
\end{equation}
with state $z(t) \in \R^r$ and parameter vector $\omega \in \mathcal{N} \subset \R^q$, where $\mathcal{N}$ is a bounded subset.
For each fixed $\omega \in \mathcal{N}$, the associated initial condition $z(0) = z_0(\omega)$ are parameterized by $\omega$ and are assumed to belong to $\mathcal{Z}_\omega \subseteq \R^r$, where $\mathcal{Z}_\omega$ is a compact invariant subset chosen as a periodic orbit for~\eqref{eq:driving_sys}.
More precisely, for fixed $\omega$, the omega-limit set of $\mathcal{Z}_\omega$ coincides with $\mathcal{Z}_\omega$ itself; see~\cite{isidori2008steady}.
We assume that the functions $s : \R^r \times \R^q \to \R^r$ and $\ell : \R^r \times \R^q \to \R^m$ are continuous in both arguments, and that $s(z, \omega)$ is locally Lipschitz in $z$, uniformly in $\omega$.
\begin{assumption} \label{ass:z_u_periodic}
	For every $\omega \in \mathcal{N}$, the state $z$ and output $u = \ell(\omega,z)$ of~\eqref{eq:driving_sys} are continuous and periodic, with period $T_\omega \in \R_+$ depending on the parameter $\omega$.
\end{assumption}
To guarantee nontriviality of the input signal, we assume that $u$ is not identically zero.
That is, there exists a $\bar{t} > 0$ such that $\ell(\omega, z(\bar{t})) \ne 0$.
To streamline the presentation, the compact interval $\T_\omega \coloneq \ps*{0, T_\omega}$, representing the time domain from $0$ to $T_\omega$ is introduced.
\begin{lemma} \label{th:u_udot_0}
	For~\eqref{eq:driving_sys} it holds that $\inner[\big]{\dot{\ell}(\omega,z)}{\ell(\omega,z)}_{\mathcal{L}_2(\T_\omega)} = 0$.
\end{lemma}
\begin{proof} \label{th:u_udot_0::proof}
	Since $u = \ell(\omega,z)$, it follows that $\dot{u} = \dot{\ell}(\omega,z)$.
	Using integration by parts, we obtain
	\begin{equation*}
		\inner*{\dot{u}}{u}_{\mathcal{L}_2(\T_\omega)} = \norm*{u(T_\omega)}^2 - \norm*{u(0)}^2 - \inner*{u}{\dot{u}}_{\mathcal{L}_2(\T_\omega)}.
	\end{equation*}
	Since $u$ is periodic with period $T_\omega$, we have that $u(T_\omega) = u(0)$.
	Thus, by symmetry of the inner product, $2 \inner*{\dot{u}}{u}_{\mathcal{L}_2(\T_\omega)} = \norm*{u(T_\omega)}^2 - \norm*{u(0)}^2 = 0$, which concludes the proof.
\end{proof}

To make the subsequent derivations well-defined, we introduce a key assumption on the interconnected system~\eqref{eq:driven_sys}-\eqref{eq:driving_sys}.

\begin{assumption} \label{ass:compact_invariant_ultimately_bounded}
	For every $\omega \in \mathcal{N}$ and every initial condition $(z(0), x(0)) \in \mathcal{Z}_\omega \times \mathcal{X}$ the subset $\mathcal{X}$ of the state space of the system~\eqref{eq:driven_sys} is such that the trajectories of the interconnected system~\eqref{eq:driven_sys}-\eqref{eq:driving_sys} are ultimately bounded, and remain in $\mathcal{Z}_\omega \times \mathcal{X}$.
	Moreover, there exists a (sufficiently) smooth mapping $\varphi : \mathcal{N} \times \R^r \to \R^n$ such that the submanifold
	\begin{equation}\label{eq:inv_manifold}
		\mathcal{M}
		\coloneq
		\set[\big]{
			(\omega, z, x) \in \mathcal{N} \times \mathcal{Z}_\omega \times \mathcal{X}
			\st
			x = \varphi(\omega, z)
		}
	\end{equation}
	is well-defined and invariant under~\eqref{eq:driven_sys}-\eqref{eq:driving_sys}.
	That is, if $(\omega, z(\bar{t}), x(\bar{t})) \in \mathcal{M}$, then $(\omega, z(t), x(t)) \in \mathcal{M}$ for all $t \ge \bar{t}$.
\end{assumption}
Under Assumption~\ref{ass:compact_invariant_ultimately_bounded}, for all $(\omega, z_0, x_0) \in \mathcal{M}$, the steady-state response of the nonlinear system~\eqref{eq:driven_sys} to the class of periodic inputs governed by the signal generator~\eqref{eq:driving_sys} is fully determined by the invariant manifold $\mathcal{M}$, see~\cite{isidori2008steady}.
Furthermore, the invariance of $\mathcal{M}$ under the flow of the interconnected system necessitates that the mapping $\varphi$ satisfies the partial differential equation (PDE), known as the \emph{invariance equation},
\begin{equation} \label{eq:pi_pde}
	f(\varphi(\omega,z),\ell(\omega,z))
	=
	\frac{\partial \varphi}{\partial z} s(\omega,z).
\end{equation}

The invariance equation~\eqref{eq:pi_pde} ensures that the vector field is tangent to the manifold $\mathcal{M}$, and is equivalent to requiring that the nonlinear system entrains to the periodic input signal~\cite{carr2012a}.
As a result, its long-term behavior can be captured entirely by the mapping $\varphi$ from the generator state to the system state.
Under suitable stability assumptions on the system~\eqref{eq:driven_sys}, the invariant manifold $\mathcal{M}$ becomes \emph{attractive}.
Specifically, if the interconnected system~\eqref{eq:driven_sys}-\eqref{eq:driving_sys} is such that all trajectories originating by initial conditions in $\mathcal{Z}_\omega \times \mathcal{X}$ converge asymptotically to $\mathcal{M}$, then it is not required for $(z(0), x(0))$ to lie exactly on $\mathcal{M}$.
This property allows the nonlinear framework to capture the asymptotic behavior of a broader class of initial conditions, much like the role of asymptotic stability in LTI systems.

Recalling that the input $u = \ell(\omega, z)$ is continuous and $T_\omega$-periodic, it follows that the state trajectories $x$ and output trajectories $y$ of the interconnected system~\eqref{eq:driven_sys}-\eqref{eq:driving_sys} when restricted to the invariant manifold $\mathcal{M}$ are likewise $T_\omega$-periodic.
The state and output restricted on $\mathcal{M}$ are given by $x = \varphi(\omega, z)$ and $y = h(\varphi(\omega, z), \ell(\omega, z))$, respectively.
Analogous to the case of asymptotically stable LTI systems, the signals $u$ and $y$ on the manifold $\mathcal{M}$ share the same fundamental period.
In this regard, frequency-domain analysis of the steady-state behavior of the nonlinear interconnected system~\eqref{eq:driven_sys}-\eqref{eq:driving_sys} may be conducted by restricting attention to the input and output signals $u$ and $y$ on the invariant manifold $\mathcal{M}$, evaluated over~$\T_\omega$.

These arguments allow defining the time-domain frequency response of the nonlinear system in~\eqref{eq:driven_sys}, when it is driven by~\eqref{eq:driving_sys} with parameter $\omega \in \mathcal{N}$, as follows.
\begin{definition}[\emph{$\omega$-response}] \label{def:w-response}
	Given the nonlinear system~\eqref{eq:driven_sys} the \emph{$\omega$-response} to the input~\eqref{eq:driving_sys} for all $\omega\in\mathcal{N}$ is given by
	\begin{equation} \label{eq:w-response}
		Y(\omega,z) \coloneq h(\varphi(\omega,z), \ell(\omega,z)).
	\end{equation}
\end{definition}
The notion of $\omega$-response is here introduced to formalize the parametric dependence of $\omega$ in the steady-state output response of the system~\eqref{eq:driven_sys} under the input signal generated by~\eqref{eq:driving_sys}.
This $\omega$-response function fully captures the effect of the class of inputs generated by~\eqref{eq:LTI-exosystem}, indicating that it can be used to analyze how the system output varies with changes in $\omega$ and in the initial conditions $z_0(\omega) \in \mathcal{Z}_\omega$.

The following example illustrates how nonlinearities affect the $\omega$-response of a system driven by an exogenous signal.
We illustrate the evolution of the $\omega$-response with respect to its generating input, emphasizing its sensitivity to variations in $\omega$ and differing initial conditions $z(0) = z_0(\omega)$.

\begin{example} \label{ex:base}
	Let~\eqref{eq:driving_sys} be the harmonic oscillator~\eqref{eq:LTI-exosystem} with parameter $\omega = (\varpi, a_u) \in \R_+^2$, state $z = (z_1, z_2) \in \R^2$, and initial condition $z_0(\omega) = (a_u, 0) \in \R^2$.
	This signal generator generates sine waves with amplitude $a_u$ and angular frequency $\varpi$.
	The nonlinear system~\eqref{eq:driven_sys} has state $x = (x_1, x_2) \in \R^2$, and its dynamics are defined as
	\begin{align}
		f(x,u)
		&=
		\pr*{
		\begin{array}{c}
			-a_1 x_1 + u \\
			-a_2 x_1^2 - a_3 x_2 + b_1 x_1 u
		\end{array}
		},
		\label{eq:nonlinear-system_example_f}
		\\
		h(x,u) &= x_1 + c_1 x_1 x_2.
		\label{eq:nonlinear-system_example_h}
	\end{align}

	The local behavior of the nonlinear system~\eqref{eq:nonlinear-system_example_f}-\eqref{eq:nonlinear-system_example_h} near the hyperbolic equilibrium $x = 0$ is an LTI system (Hartman-Grobman theorem~\cite{hartman1982ordinary}) that can be obtained by setting $(a_2, b_1, c_1) = (0, 0, 0)$.
	To compute the $\omega$-response associated with the interconnected system~\eqref{eq:LTI-exosystem}-\eqref{eq:nonlinear-system_example_f}, with respect to the output~\eqref{eq:nonlinear-system_example_h}, we seek a mapping $\varphi(\omega, z) = (\varphi_1, \varphi_2)$ that satisfies the invariance equation~\eqref{eq:pi_pde}.
	In this case, the interconnected system~\eqref{eq:LTI-exosystem}-\eqref{eq:nonlinear-system_example_f} yields the following quadratic solution
	\begin{equation} \label{eq:nonlinear-system_example_phi}
		\varphi(\omega, z)
		=
		\pr*{
		\begin{array}{c}
			d_{1,\omega} z_1 + d_{2,\omega} z_2 \\
			d_{3,\omega} z_1 z_2 + d_{4,\omega} z_1^2 + d_{5,\omega} z_2^2
		\end{array}
		},
	\end{equation}
	where the coefficients $d_{i,\omega}$, $i \in \pc*{1, \ldots, 5}$, are functions of the parameter $\omega$ and determined by solving the invariance equation.
	With the mapping~\eqref{eq:nonlinear-system_example_phi} in hand, we can evaluate the $\omega$-response using the composition $h \circ \varphi$, as defined in~\eqref{eq:w-response}.

	To illustrate the sensitivity of the $\omega$-response, we consider the specific parameter configuration
	\begin{align} \label{eq:parameters_ex1}
		a_1 & = 0.5, &
		a_2 & = 1, &
		a_3 & = 1, &
		b_1 & = 1, &
		c_1 & = 1.
	\end{align}
	Figure~\ref{fig:ex:base:Y_planes} shows the surface of the $\omega$-response for three different input frequencies, \emph{i.e.}, $\omega = (0.01, 1)$, $\omega = (0.1, 1)$, and $\omega = (1, 1)$.
	The surface exhibits strongly nonlinear behavior for small values of $\varpi$, as evidenced by its curvature.
	As $\varpi$ increases, the surface gradually flattens and approaches a planar shape, indicating convergence toward behavior typical of a linear system.
	This trend suggests that, for large $\varpi$, the nonlinear effects become negligible, and the system behaves similarly to its linear approximation.

	Figure~\ref{fig:ex:base:ts_U} displays representative input signals $z_1$, and Figure~\ref{fig:ex:base:ts_Y} shows the corresponding $\omega$-response trajectories as a function of $a_u$ and fixed $\varpi = 0.1$.
	For small amplitudes $a_u$, the $\omega$-response closely tracks the input signal, indicating strong signal correlation and validation of the Hartman-Grobman theorem around the origin.
	However, as the input amplitude approaches $0.8$, the response exhibits deviation from the input trajectory.
	For a small $\varpi$, such as $\omega = 0.1$, the output shows both phase inversion and amplitude distortion, with changes in trajectory shape.
	These phenomena reflect the nonlinear system sensitivity to input magnitude and highlight how the $\omega$-response captures such effects. Therefore, despite its stability, the system performs differently for different excitations. 

	The framework presented in Section~\ref{sec:nonlinear-freq-analysis} unifies the analysis of all these phenomena and extends them through a suitable frequency response function for nonlinear systems.
\end{example}
\begin{figure*}[t]
	\centering
	\condVSpace{-1.5em}
	\begin{subfigure}[b]{0.325\linewidth}
	\centering
	\includegraphics[width=\linewidth]{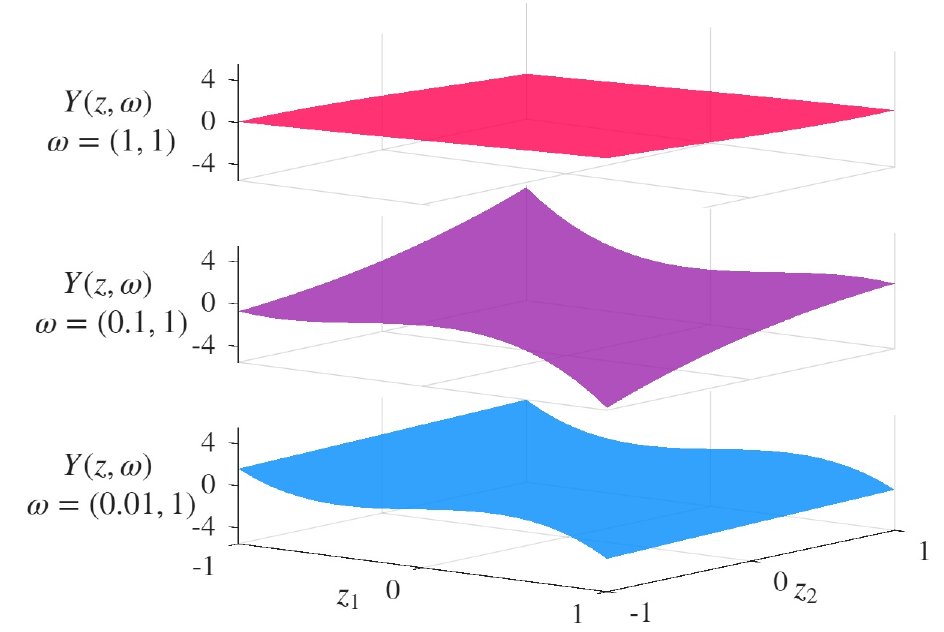}
	\condVSpace{-1.5em}
	\caption{ }
	\label{fig:ex:base:Y_planes}
	\end{subfigure}
	\hfill
	\begin{subfigure}[b]{0.325\linewidth}
	\centering
	\includegraphics[width=\linewidth]{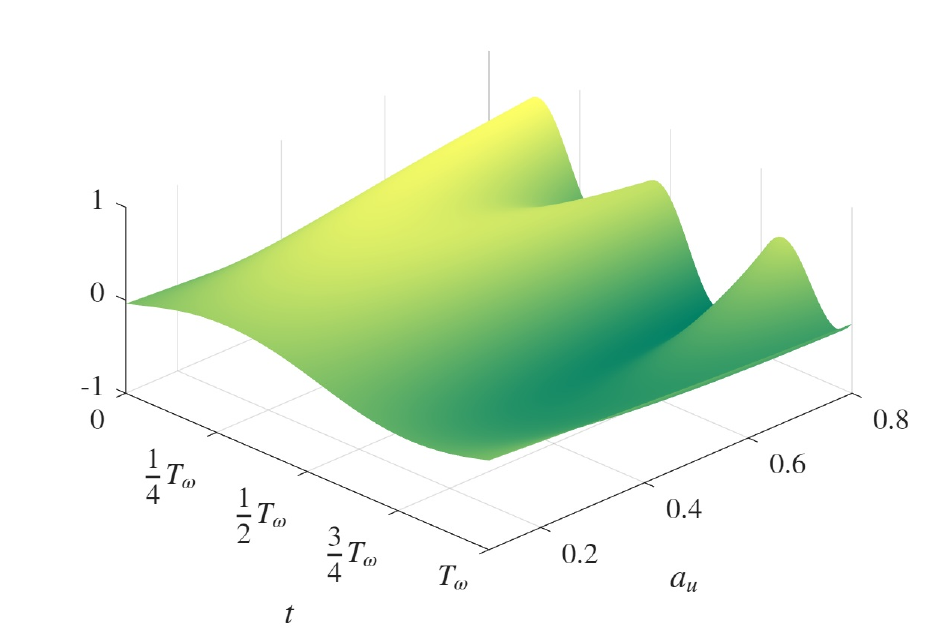}
	\condVSpace{-1.5em}
	\caption{ }
	\label{fig:ex:base:ts_Y}
	\end{subfigure}
	\hfill
	\begin{subfigure}[b]{0.325\linewidth}
	\centering
	\includegraphics[width=\linewidth]{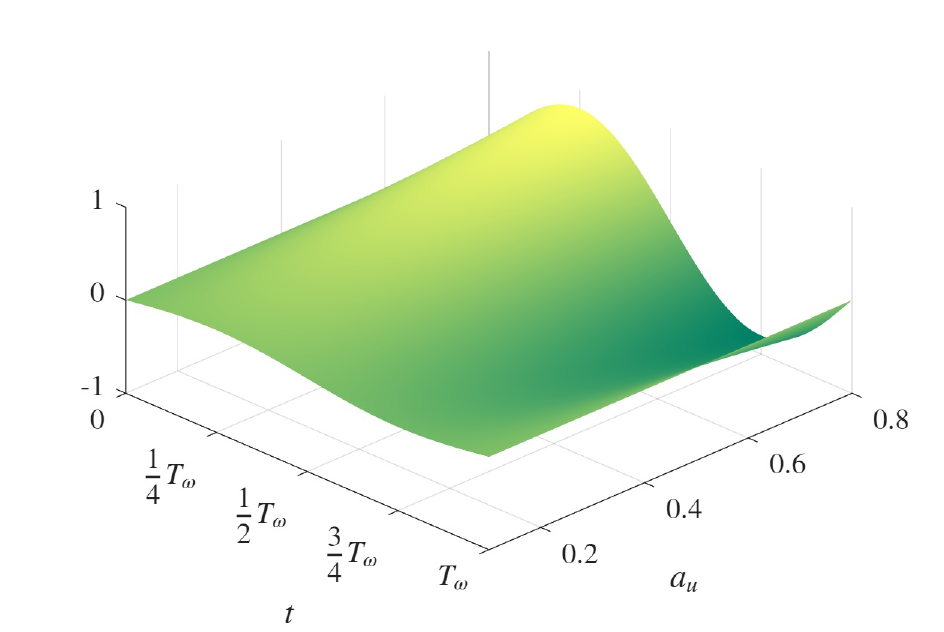}
	\condVSpace{-1.5em}
	\caption{ }
	\label{fig:ex:base:ts_U}
	\end{subfigure}
	\caption{
		Graphical analysis of the $\omega$-response of the nonlinear system~\eqref{eq:nonlinear-system_example_f}-\eqref{eq:nonlinear-system_example_h} to the input defined in~\eqref{eq:LTI-exosystem} corresponding to the parameter configuration~\eqref{eq:parameters_ex1}:
		(a) Surface of $Y(\omega, z)$ with $z \in [-1, 1]^2$ with $a_u = 1$ and $\varpi \in \pc*{10^{-2}, 10^{-1}, 1}$;
		(b) Time evolution of the function $Y(\omega, z(t))$ over the domain $t \in \T_\omega$, with $\varpi = 0.1$ and $a_u \in [0.1, 0.8]$;
		(c) Time evolution of the	input $\ell(\omega, z(t))$ under the same conditions as in (b).
	}
	\label{fig:nonlinear_w_response}
	\condVSpace{-1.5em}
\end{figure*}

\vspace{-1em}
\section{The Concept of Frequency Response of Nonlinear Systems}

\label{sec:nonlinear-freq-analysis}

We recall that the continuity and periodicity of $\ell(\omega, z)$ and $Y(\omega, z)$, with period $T_\omega \in \R_+$ depending on the parameter $\omega$, imply that both functions are bounded and square-integrable, and thus belong to well-defined function spaces.
Specifically, boundedness ensures finite amplitude over the domain, while square integrability guarantees finite energy in the $\mathcal{L}_2$ sense.

With this in mind, to characterize the behavior of the nonlinear system~\eqref{eq:driven_sys} with respect to the input~\eqref{eq:driving_sys} through the $\omega$-response, we define a series of functions which are instrumental to characterize the frequency response function.
\begin{definition}[\emph{$\omega$-gain}] \label{def:w_gain}
	The \emph{$\omega$-gain} function $\alpha : \mathcal{N} \to \R$ of the system~\eqref{eq:driven_sys} to the input~\eqref{eq:driving_sys} for all $\omega\in\mathcal{N}$ is given by
	\begin{equation}\label{eq:def_w-gain}
		\alpha(\omega)
		\coloneq
		\frac
		{ \norm[\big]{Y(\omega,z)}_{\mathcal{L}_2 \pr*{\T_\omega}} }
		{ \norm[\big]{\ell(\omega,z)}_{\mathcal{L}_2 \pr*{\T_\omega}} }
		.
	\end{equation}
\end{definition}
The $\omega$-gain function\footnote{
This gain function was used in~\cite{isidori1992disturbance} in the context of $H_\infty$-control, with the difference that here the $\mathcal{L}_2$ norm is defined over the compact $\T_\omega$. In~\cite{moreschini2026onfrequency}, a similar gain was given in terms of $\mathcal{L}_p$ norm.
}, as defined in~\eqref{eq:def_w-gain}, quantifies the amplification (or attenuation) of the input signal at every frequency \diff{associated with} $\omega$.
It describes the ratio between the energy (\emph{i.e.}, the $\mathcal{L}_2$ norm) of the $\omega$-response and that of the reference input signal $\ell(\omega, z)$, both measured over one period corresponding to $\omega$.
However, this function possesses a notable property, which is proved in the following statement.

\begin{proposition} \label{th:LTI-gain}
	Consider the LTI system~\eqref{eq:LTI-system} with transfer function $H : \C \to \C$, driven by the harmonic signal generator~\eqref{eq:LTI-exosystem} with parameter $\omega = (\varpi, a_u, \phi_u) \in \R_+^2 \times \R$ and initial condition $z_0(\omega) = (a_u \sin\phi_u, a_u \cos\phi_u) \in \R^2$.
	Then, the associated $\omega$-gain function~\eqref{eq:def_w-gain} satisfies
	\begin{equation} \label{eq:w-gain-linear}
		\alpha(\omega) = \norm*{ H\pr*{j \varpi} }.
	\end{equation}
\end{proposition}
\begin{proof} \label{th:LTI-gain::proof}
	The state of the harmonic input with initial condition $z_0(\omega) = (a_u \sin\phi_u, a_u \cos\phi_u)$ and parameter $\omega = (\varpi, a_u, \phi_u)$ is given by $z_1(t) = a_u \sin(\varpi t + \phi_u)$ and $z_2(t) = a_u \cos(\varpi t + \phi_u)$ with period $T_\omega = \frac{2\pi}{\varpi}$.
	Thus, we also have $\ell(\omega, z) = z_1$, $\dot{\ell}(\omega, z) = \varpi z_2$ and, as shown in Lemma~\ref{th:u_udot_0},
	\begin{equation} \label{eq:LTI:<z1,z2>=0}
		\inner{z_1}{z_2}_{\mathcal{L}_2 \pr*{\T_\omega}}
		=
		\varpi^{-1} \inner[\big]{\ell(\omega, z)}{\dot{\ell}(\omega, z)}_{\mathcal{L}_2 \pr*{\T_\omega}}
		=
		0.
	\end{equation}
	Therefore, the input $\ell(\omega, z) = z_1$ has norm
	\begin{equation*}
		\norm*{\ell(\omega, z)}_{\mathcal{L}_2(\T_\omega)}
		=
		a_u \sqrt{\frac{\pi}{\varpi}}.
	\end{equation*}
	Since~\eqref{eq:driven_sys} is also an LTI system of the form~\eqref{eq:LTI-system} we have that $Y(\omega, z) = C\Phi_1(\varpi) z_1 + C\Phi_2(\varpi) z_2$, where $\Phi(\varpi) \coloneq [\Phi_1(\varpi), \Phi_2(\varpi)] \in \R^{n \times 2}$ satisfies the Sylvester equation~\eqref{eq:Sylvester}.
	From Lemma~\ref{th:LTI-tf-Phi}, recall that $C \Phi_1(\varpi) + j C \Phi_2(\varpi) = H(j\varpi)$.
	Then, since~\eqref{eq:LTI:<z1,z2>=0} holds, the $\mathcal{L}_2$-norm of $Y(\omega, z)$ over $\T_\omega$ yields
	\begin{align}
		\norm[\big]{Y(\omega, z)}_{\mathcal{L}_2 \pr*{\T_\omega}}
		&
		= \norm[\big]{C \Phi_1(\varpi)z_1 + C \Phi_2(\varpi)z_2}_{\mathcal{L}_2 \pr*{\T_\omega}}
		\nonumber
		\\
		&
		= a_u \sqrt{\frac{\pi}{\varpi}} \sqrt{\pr*{C \Phi_1(\varpi)}^2 + \pr*{C \Phi_2(\varpi)}^2}
		\nonumber
		\\
		&
		= a_u \sqrt{\frac{\pi}{\varpi}} \norm{H(j\varpi)},
		\label{eq:LTI_normY}
	\end{align}
	and thus the $\omega$-gain function~\eqref{eq:def_w-gain} coincides with the amplitude change induced by its frequency response $H(j\varpi)$, that is~\eqref{eq:w-gain-linear}.
\end{proof}
\begin{definition}[\emph{$\omega$-phase}] \label{def:w_phase}
	The \emph{$\omega$-phase} function $\vartheta : \mathcal{N} \to \R$ of the system~\eqref{eq:driven_sys} to the input~\eqref{eq:driving_sys} for all $\omega\in\mathcal{N}$ is given by
	\begin{equation}\label{eq:def_w-phase}
		\vartheta(\omega)
		\coloneq
		\arg(\mathfrak{Re}(\omega) + j\mathfrak{Im}(\omega))
	\end{equation}
	where
	\begin{align*}
		\mathfrak{Re}(\omega)
		&
		\coloneq
			\frac{
				\inner[\big]{\ell(\omega,z)}{Y(\omega,z)}_{\mathcal{L}_2 \pr*{\T_\omega}}
			}{
				\norm[\big]{\ell(\omega,z)}_{\mathcal{L}_2 \pr*{\T_\omega}}\norm[\big]{Y(\omega,z)}_{\mathcal{L}_2 \pr*{\T_\omega}}
			},
		\\
		\mathfrak{Im}(\omega)
		&
		\coloneq
			\frac{
				\inner[\big]{\dot{\ell}(\omega,z)}{Y(\omega,z)}_{\mathcal{L}_2 \pr*{\T_\omega}}
			}{
				\norm[\big]{\dot{\ell}(\omega,z)}_{\mathcal{L}_2 \pr*{\T_\omega}}\norm[\big]{Y(\omega,z)}_{\mathcal{L}_2 \pr*{\T_\omega}}
			}.
	\end{align*}
\end{definition}
The notion of $\omega$-phase function as defined in~\eqref{eq:def_w-phase} provides a geometric characterization of the phase angle between the $\omega$-response and the subspace spanned by the functions $\ell(\omega,z)$ and $\dot{\ell}(\omega,z)$ at all $\omega$.
Specifically, the $\omega$-phase function $\vartheta(\omega)$ is defined as the argument (angle) of the resulting complex-valued projection where the real part $\mathfrak{Re}(\omega)$ corresponds to the inner product with the input $\ell(z, \omega)$, capturing the in-phase component, and the imaginary part $\mathfrak{Im}(\omega)$ corresponds to the inner product with the time derivative $\dot{\ell}(z, \omega)$, capturing the orthogonal component.
Since $\ell(\omega,z)$ and $\dot{\ell}(\omega,z)$ form an orthogonal pair (see Lemma~\ref{th:u_udot_0}), the functions $\mathfrak{Im}(\omega)$ and $\mathfrak{Re}(\omega)$ form an orthogonal basis\footnote{A similar argument can be obtained applying the Hilbert transform on the input signal $u(t)$, which produces an orthogonal signal~\cite{king2009hilbert}.} 
of a subspace of the function space $\mathcal{L}_2 \pr*{\T_\omega}$.
This orthogonality indicates that they represent independent components of $Y(\omega,z)$, guaranteeing that $\vartheta(\omega)$ is a well-defined phase angle without cross-term interference.
Consequently, the sign and magnitude determine the direction and the exact value of the $\omega$-phase in the complex plane spanned by the function $\ell(\omega,z)$ and $\dot{\ell}(\omega,z)$.
\begin{proposition} \label{th:LTI-phase}
	Consider the LTI system~\eqref{eq:LTI-system} with transfer function $H : \C \to \C$, driven by the harmonic signal generator~\eqref{eq:LTI-exosystem} with parameter $\omega = (\varpi, a_u, \phi_u) \in \R_+^2 \times \R$ and initial condition $z_0(\omega) = (a_u \sin\phi_u, a_u \cos\phi_u) \in \R^2$.
	Then, the associated $\omega$-phase function~\eqref{eq:def_w-phase} satisfies
	\begin{equation} \label{eq:w-phase-linear}
		\vartheta(\omega) = \arg \pr*{ H(j\omega) }.
	\end{equation}
\end{proposition}
\begin{proof} \label{th:LTI-phase::proof}
	As shown at the beginning of the proof of Proposition~\ref{th:LTI-gain}, in these settings, we have $T_\omega = \frac{2\pi}{\varpi}$, $\ell(\omega,z) = z_1$, $\dot{\ell}(\omega,z) = \varpi z_2$, and $\inner*{z_1}{z_2}_{\mathcal{L}_2 \pr*{\T_\omega}} = 0$.
	Additionally, the $\mathcal{L}_2$-norm of $\ell(\omega,z)$ and $\dot{\ell}(\omega,z)$ over $\T_\omega$ satisfy
	\begin{align*}
		\norm[\big]{\dot{\ell}(\omega,z)}_{\mathcal{L}_2 \pr*{\T_\omega}}
		&
		=
		a_u \sqrt{\frac{\pi}{\varpi}},
		&
		\norm[\big]{\dot{\ell}(\omega,z)}_{\mathcal{L}_2 \pr*{\T_\omega}}
		&
		=
		a_u \sqrt{\pi \varpi}.
	\end{align*}
	Now, recall that, since the system~\eqref{eq:driven_sys} is LTI, $Y(\omega, z) = C\Phi_1(\varpi) z_1 + C\Phi_2(\varpi) z_2$, where $\Phi(\varpi) \coloneq [\Phi_1(\varpi), \Phi_2(\varpi)] \in \R^{n \times 2}$ satisfies the Sylvester equation~\eqref{eq:Sylvester}.
	Then, we have
	\begin{align*}
		\inner[\big]{\ell(\omega,z)}{Y(\omega,z)}_{\mathcal{L}_2 \pr*{\T_\omega}}
		\hspace{-8ex}
		&
		\\
		&
		=
		\inner[\big]{z_1}{C\Phi_1(\varpi) z_1 + C\Phi_2(\varpi) z_2 }_{\mathcal{L}_2 \pr*{\T_\omega}}
		\\[-0.25em]
		&
		=
		C\Phi_1(\varpi) \inner[\big]{z_1}{z_1}_{\mathcal{L}_2 \pr*{\T_\omega}}
		=
		C\Phi_1(\varpi) \frac{\pi a_u^2}{\varpi},
	\end{align*}
	and similarly,
	\begin{align*}
		\inner[\big]{\dot{\ell}(\omega,z)}{Y(\omega,z)}_{\mathcal{L}_2 \pr*{\T_\omega}}
		\hspace{-8ex}
		&
		\\
		&
		=
		\inner[\big]{\varpi z_2}{C\Phi_1(\varpi) z_1 + C\Phi_2(\varpi) z_2 }_{\mathcal{L}_2 \pr*{\T_\omega}}
		\\
		&
		=
		\varpi C\Phi_2(\varpi) \inner[\big]{z_2}{z_2}_{\mathcal{L}_2 \pr*{\T_\omega}}
		=
		C\Phi_2(\varpi) \pi a_u^2.
	\end{align*}
	Then, by definition of $\mathfrak{Re}(\omega)$ and $\mathfrak{Im}(\omega)$, it follows that
	\begin{align*}
		\mathfrak{Re}(\omega)
		&
		= \frac{C \Phi_1(\varpi) a_u}{\norm[\big]{Y(\omega,z)}_{\mathcal{L}_2 \pr*{\T_\omega}}} \sqrt{\frac{\pi}{\varpi}},
		\\
		\mathfrak{Im}(\omega)
		&
		= \frac{ C \Phi_2(\varpi) a_u}{\norm[\big]{Y(\omega,z)}_{\mathcal{L}_2 \pr*{\T_\omega}}} \sqrt{\frac{\pi}{\varpi}}.
	\end{align*}
	Therefore, we obtain
	\begin{align}
		\mathfrak{Re}(\omega) + j \mathfrak{Im}(\omega)
		&
		=
		a_u
		\sqrt{\frac{\pi}{\varpi}}
		\frac
		{ C \Phi_1(\varpi) + j C \Phi_1(\varpi) }
		{ \norm[\big]{Y(\omega,z)}_{\mathcal{L}_2 \pr*{\T_\omega}}}
		\nonumber
		\\
		&
		=
		a_u
		\sqrt{\frac{\pi}{\varpi}}
		\frac
		{ H(j \varpi) }
		{ \norm[\big]{Y(\omega,z)}_{\mathcal{L}_2 \pr*{\T_\omega}} }
		\label{eq:LTI_ReIm}
	\end{align}
	Hence, the $\omega$-phase of the LTI system in response to a harmonic input coincides with the phase shift introduced by its frequency response $H(j\omega)$, \emph{i.e.},
	\begin{equation*}
		\vartheta(\omega)
		=
		\arg \pr*{ \mathfrak{Re}(\omega) + j \mathfrak{Im}(\omega) }
		=
		\arg \pr*{ H(j\diff{\varpi}) },
	\end{equation*}
	that is~\eqref{eq:w-phase-linear}.
\end{proof}
\begin{definition}[\emph{$\omega$-radius}] \label{def:w_radius}
	The \emph{$\omega$-radius} function $r : \mathcal{N} \to \R$ of the system~\eqref{eq:driven_sys} to the input~\eqref{eq:driving_sys} for all $\omega\in\mathcal{N}$ is given by
	\begin{equation} \label{eq:def_w-radius}
		r(\omega)
		\coloneq
		\norm{ \mathfrak{Re}(\omega) + j \mathfrak{Im}(\omega) }.
	\end{equation}
\end{definition}
Before discussing the meaning of the $\omega$-radius, let us first introduce the following theorem, which proves important properties of this function.
\begin{theorem} \label{th:w-radius}
	For~\eqref{eq:def_w-radius}, it holds that:
	\begin{enumerate}[label=\textbf{\textit{(\roman*)}}, ref=\textbf{\textit{(\roman*)}}]
		\item \label{th:w-radius::r01} $r(\omega) \in (0, 1]$,
		\item \label{th:w-radius::fraction} the squared $\omega$-radius is equal to
		\begin{equation*}
			r(\omega)^2
			=
			1 -
			\frac
			{ \norm[\big]{ Y(\omega,z) - \overline{Y}(\omega,z) }_{\diff{\mathcal{L}_2 \pr*{\T_\omega}}}^2 }
			{ \norm[\big]{ Y(\omega,z)}_{\mathcal{L}_2 \pr*{\T_\omega}}^2 } ,
		\end{equation*}
		where $\overline{Y}(\omega,z)$ is the orthogonal projection of $Y(\omega,z)$ onto the subspace $\overline{\mathcal{L}_2} \pr*{\T_\omega}$ of $\mathcal{L}_2 \pr*{\T_\omega}$ that spans $\ell(\omega,z)$ and $\dot{\ell}(\omega,z)$.
		\item \label{th:w-radius::r=1} $r(\omega) = 1$ if and only if $Y(\omega,z) = \overline{Y}(\omega,z)$.
	\end{enumerate}

\end{theorem}
\begin{proof} \label{th:w-radius::proof}
	For brevity, let us recall that $u = \ell(\omega,z)$ and $\dot{u} = \dot{\ell}(\omega,z)$.
	For the definition of $\overline{Y}(\omega,z)$, we know that
	\begin{equation*}
		\overline{Y}(\omega,z)
		=
		\lambda_1 u
		+
		\lambda_2 \dot{u},
	\end{equation*}
	where
	\begin{align*}
		\lambda_1
		&
		=
		\frac
		{ \inner*{u}{Y(\omega,z)}_{\mathcal{L}_2 \pr*{\T_\omega}} }
		{ \norm*{u}_{\mathcal{L}_2 \pr*{\T_\omega}}^2 },
		&
		\lambda_2
		&
		=
		\frac
		{ \inner*{\dot{u}}{Y(\omega,z)}_{\mathcal{L}_2 \pr*{\T_\omega}} }
		{ \norm*{\dot{u}}_{\mathcal{L}_2 \pr*{\T_\omega}}^2 }.
	\end{align*}
	Then, recalling that $\inner*{u}{\dot{u}}_{\mathcal{L}_2(\T_\omega)} = 0$ (see Lemma~\ref{th:u_udot_0}), we have
	\begin{align*}
		\norm[\big]{ Y(\omega,z) - \overline{Y}(\omega,z) }_{\diff{\mathcal{L}_2 \pr*{\T_\omega}}}^2
		=
		\hspace{-15ex}
		&
		\\
		&
		\norm[\big]{ Y(\omega,z) }_{\mathcal{L}_2 \pr*{\T_\omega}}^2
		+
		\lambda_1^2 \norm[\big]{ u }_{\mathcal{L}_2 \pr*{\T_\omega}}^2
		+
		\lambda_2^2 \norm[\big]{ \dot{u} }_{\mathcal{L}_2 \pr*{\T_\omega}}^2
		\\
		&
		-
		2 \lambda_1
		\inner[\big]{ Y(\omega,z) }{u}_{\mathcal{L}_2 \pr*{\T_\omega}}
		-
		2 \lambda_2
		\inner[\big]{ Y(\omega,z) }{\dot{u}}_{\mathcal{L}_2 \pr*{\T_\omega}} .
	\end{align*}
	From the definitions of $\lambda_1$, $\lambda_2$, $\mathfrak{Re}(\omega)$, and $\mathfrak{Im}(\omega)$, we have
	\begin{align*}
		\lambda_1^2 \norm[\big]{ u }_{\mathcal{L}_2 \pr*{\T_\omega}}^2
		&
		\! =
		\lambda_1
		\inner[\big]{ Y(\omega,z) }{u}_{\mathcal{L}_2 \pr*{\T_\omega}}
		\! =
		\norm[\big]{ Y(\omega,z) }_{\mathcal{L}_2 \pr*{\T_\omega}}^2
		\mathfrak{Re}(\omega)^2 \! ,
		\\
		\lambda_2^2 \norm[\big]{ \dot{u} }_{\mathcal{L}_2 \pr*{\T_\omega}}^2
		&
		\! =
		\lambda_2
		\inner[\big]{ Y(\omega,z) }{\dot{u}}_{\mathcal{L}_2 \pr*{\T_\omega}}
		\! =
		\norm[\big]{ Y(\omega,z) }_{\mathcal{L}_2 \pr*{\T_\omega}}^2
		\mathfrak{Im}(\omega)^2 \! .
	\end{align*}
	Thus, since $r(\omega)^2 = \mathfrak{Re}(\omega)^2 + \mathfrak{Im}(\omega)^2$, we have
	\begin{align}
		\norm[\big]{
			Y(\omega,z)
			- \lambda_1 u
			- \lambda_2 \dot{u}
		}_{\mathcal{L}_2 \pr*{\T_\omega}}^2
		\nonumber
		\hspace{-25ex}
		&
		\\
		&
		=\norm[\big]{ Y(\omega,z) }_{\mathcal{L}_2 \pr*{\T_\omega}}^2
		\pr[\big]{1 - \mathfrak{Re}(\omega)^2 - \mathfrak{Im}(\omega)^2}
		\nonumber
		\\
		&
		=\norm[\big]{ Y(\omega,z) }_{\mathcal{L}_2 \pr*{\T_\omega}}^2
		\pr[\big]{1 - r(\omega)^2} .
		\label{th:w-radius::proof:eq:normY}
	\end{align}
	Recall that $r(\omega) \ge 0$ by definition; however, since~\eqref{th:w-radius::proof:eq:normY} holds, then $r(\omega) = 0$ if and only if $ u(t) = 0$ for all $t\in\T_\omega$, which is not true by assumption.
	Hence, $r(\omega) > 0$.
	Moreover, we conclude that~\eqref{th:w-radius::proof:eq:normY} holds only for $1 - r(\omega)^2 \ge 0$, which implies that $r(\omega) \le 1$, proving statement~\ref{th:w-radius::r01}.
	Statement~\ref{th:w-radius::fraction} is obtained by manipulating~\eqref{th:w-radius::proof:eq:normY}.
	Finally,~\ref{th:w-radius::fraction} readily implies statement~\ref{th:w-radius::r=1}, because the $\mathcal{L}_2 \pr*{\T_\omega}$ norm of a continuous signal is equal to $0$ if and only if the signal is identically~$0$.
\end{proof}
\begin{figure}[t]
\condVSpace{-1.5em}
\centering
\tdplotsetmaincoords{75}{120}
\begin{tikzpicture}[tdplot_main_coords]

	\coordinate (O) at (0,0,0);
	\coordinate (u) at (4.75,0,0);
	\coordinate (h) at (0,5,0);

	\pgfmathsetmacro{\Yx}{4}
	\pgfmathsetmacro{\Yy}{3}
	\pgfmathsetmacro{\Yz}{3.5}

	\coordinate (Y) at (\Yx,\Yy,\Yz);
	\coordinate (Yb) at (\Yx,\Yy,0);
	\coordinate (Re) at (\Yx,0,0);
	\coordinate (Im) at (0,\Yy,0);

	\coordinate (ur) at (-1,0,0);
	\coordinate (hr) at (0,-1,0);

	\path[draw, dashed] (Yb) -- (Re);
	\path[draw, dashed] (Yb) -- (Im);

	\path[draw, dashed] (O) -- (ur);
	\path[draw, dashed] (O) -- pic [draw, black, fill=black!20, solid, thick, angle radius=0.35cm, angle eccentricity=1.6] {right angle = ur--O--hr} (hr);

	\path[draw, ->, thick] (O) -- node[pos=1, anchor=north]{$\ell(\omega, z)$} (u);
	\path[draw, ->, thick] (O) -- node[pos=1, anchor=north]{$\dot{\ell}(\omega, z)$} (h);

	\path[draw, ->, ultra thick, BrickRed] (O) -- node[pos=1, anchor=south]{$Y(\omega, z)$} (Y);

	\path[draw, ->, ultra thick, NavyBlue] (O) --
	node[pos=1, anchor=north]{$\overline{Y}(\omega, z)$}
	pic [draw, ->, red, thick, angle radius=0.35cm, pic text={$\vartheta(\omega)$}, angle eccentricity=1.6] {angle = u--O--Yb}
	(Yb)
	;

	\path[draw, ->, ultra thick, PineGreen] (O) -- node[pos=0.8, anchor=south east]{$\lambda_1 \ell(\omega, z)$} (Re);
	\path[draw, ->, ultra thick, Plum] (O) -- node[pos=1, anchor=south]{$\lambda_2 \dot{\ell}(\omega, z)$} (Im);

	\path[draw, dashed, very thick, Violet] (Y) -- node[pos=0.25, anchor=west]{\diff{$\norm*{Y(\omega, z)}_{\mathcal{L}_2 \pr*{\T_\omega}}\sqrt{1 - r(\omega)^2}$}} (Yb);

\end{tikzpicture}
\caption{
	Illustration of the space $\overline{\mathcal{L}_2} \pr*{\T_\omega}$ that spans $\ell(\omega, z)$ and $\dot{\ell}(\omega, z)$ in relation with $Y(\omega, z)$, $\overline{Y}(\omega, z)$, the $\vartheta(\omega)$, and $r(\omega)$.
}
\condVSpace{-1.5em}
\label{fig:3dvec}
\end{figure}
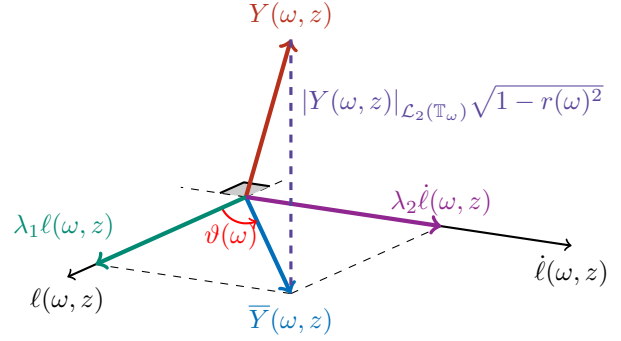
As represented in Figure~\ref{fig:3dvec}, by virtue of Theorem~\ref{th:w-radius}, the $\omega$-radius shows the fraction of the norm of the $\omega$-response that lies on the finite dimensional subspace $\overline{\mathcal{L}_2} \pr*{\T_\omega}$ of $\mathcal{L}_2 \pr*{\T_\omega}$ that spans the input $\ell(\omega,z)$ and its derivative $\dot{\ell}(\omega,z)$.
Therefore, the $\omega$-radius function takes values between two bounds:
$r(\omega) = 1$ when the $\omega$-response function is an element of $\overline{\mathcal{L}_2} \pr*{\T_\omega}$,
while $r(\omega) = 0$ when the $\omega$-response is orthogonal to both the input and its derivative, which, as stated in Theorem~\ref{th:w-radius}, is impossible.
Hence, the $\omega$-radius function can be viewed as a quantitative measure of the distortion of the input induced by the system.
Indeed, when $r(\omega) = 1$, the $\omega$-response function differs from the input only in amplitude, as described by the $\omega$-gain, and in phase lag, as described by the $\omega$-phase.
In the case of a sine wave input (generated from a harmonic oscillator), $r(\omega) = 1$ implies that the output is again a sine wave, with a different amplitude and phase.
An additional interpretation of $r(\omega) = 1$ is also provided in Lemma~\ref{th:r1_convolution} in Appendix.
Hence, $r(\omega)$ in~\eqref{eq:def_w-radius} can thus be used to determine classes of systems that preserve the input class in the output, as studied in~\cite{chaffey2025amplitude}.

An example of a system that satisfies this property is an LTI system fed with a harmonic oscillator, as shown in the following proposition.
\begin{proposition} \label{th:LTI-radius}
	Consider the LTI system~\eqref{eq:LTI-system} with transfer function $H : \C \to \C$, driven by the harmonic signal generator~\eqref{eq:LTI-exosystem} with parameter $\omega = (\varpi, a_u, \phi_u) \in \R_+^2 \times \R$ and initial condition $z_0(\omega) = (a_u \sin\phi_u, a_u \cos\phi_u) \in \R^2$.
	Then, the associated $\omega$-radius function~\eqref{eq:def_w-radius} satisfies
	\begin{equation} \label{eq:w-radius-linear}
		r(\omega) = 1.
	\end{equation}
\end{proposition}
\begin{proof}
	Recalling~\eqref{eq:LTI_ReIm} in the proof of Proposition~\ref{th:LTI-phase}, we have
	\begin{align*}
		r(\omega)
		=
		\norm{
			\mathfrak{Re}(\omega)
			+
			j\mathfrak{Im}(\omega)
		}
		=
		a_u
		\sqrt{\frac{\pi}{\varpi}}
		\frac
		{ \norm*{H(j \varpi)} }
		{ \norm[\big]{Y(\omega,z)}_{\mathcal{L}_2 \pr*{\T_\omega}} }
		.
	\end{align*}
	Then, recalling~\eqref{eq:LTI_normY} in the proof of Proposition~\ref{th:LTI-phase}, we obtain
	\begin{equation*}
		r(\omega)
		=
		a_u
		\sqrt{\frac{\pi}{\varpi}}
		\frac
		{ \norm*{H(j \varpi)} }
		{ a_u \sqrt{\frac{\pi}{\varpi}} \norm{H(j\varpi)} }
		=
		1,
	\end{equation*}
	which concludes the proof.
\end{proof}
Using the rationale behind the definition of frequency response function, we can employ the notion of $\omega$-gain, $\omega$-phase, and $\omega$-radius to define the nonlinear frequency response of the system~\eqref{eq:driven_sys} to the input~\eqref{eq:driving_sys} as follows.
\begin{definition}[\emph{Frequency response}]
	The \emph{frequency response} function $\Gamma : \mathcal{N} \to \C$ of the system~\eqref{eq:driven_sys} to the input~\eqref{eq:driving_sys} for all $\omega\in\mathcal{N}$ is given by
	\begin{equation} \label{eq:def_freq_resp}
			\Gamma(\omega) \coloneq \alpha(\omega) r(\omega) e^{j \vartheta(\omega)}.
	\end{equation}
\end{definition}
The nonlinear frequency response function~\eqref{eq:def_freq_resp} takes the standard phasor form, expressed in terms of the functions $\omega$-gain, $\omega$-phase, and $\omega$-radius associated with the nonlinear system~\eqref{eq:driven_sys} and the nonlinear input~\eqref{eq:driving_sys}. In the general nonlinear case, the $\omega$-radius function plays a role in balancing the relation between the cosine and sine of the $\omega$-phase and the Euler identity. Specifically, for a given $\vartheta(\omega)$, the sine and cosine satisfy
\begin{align} \label{eq:cos_sin}
	\cos(\vartheta(\omega))
	&
	=
	\frac{ \mathfrak{Re}(\omega) }{ r(\omega) } \ ,
	&
	\sin(\vartheta(\omega))
	&
	= \frac{\mathfrak{Im}(\omega)}{r(\omega)} \ ,
\end{align}
respectively.
Hence, the frequency response function $\Gamma(\omega)$ can be expressed by Euler's formula in the form
\begin{equation} \label{eq:tiangular_frequency}
\begin{aligned}
	\Gamma(\omega)
	&
	=
	\alpha(\omega) r(\omega) \cos(\vartheta(\omega))
	+
	j \alpha(\omega) r(\omega) \sin(\vartheta(\omega))
	\\
	&
	=
	\alpha(\omega)\mathfrak{Re}(\omega)
	+
	j
	\alpha(\omega) \mathfrak{Im}(\omega) \ .
\end{aligned}
\end{equation}

A geometric representation of the frequency response function is illustrated in Figure~\ref{fig:3dvec}.
It should also be noted that the frequency response function satisfies
\begin{equation} \label{eq:def_mag_freq_resp}
	\norm{\Gamma(\omega)} = \norm{\alpha(\omega) r(\omega)} \diff{ = \alpha(\omega) r(\omega) },
\end{equation}
which emphasizes that the magnitude of the frequency response differs from its gain.
This is because the magnitude depends not only on the $\omega$-gain function, but also on the effect of the $\omega$-radius, \diff{which may attenuates it as $r(\omega) \in (0, 1]$ (see Point~\ref{th:w-radius::r=1} of Theorem~\ref{th:w-radius})}.
However, when the input class is preserved by the system, \emph{i.e.}, \diff{$r(\omega)=1$}, the magnitude and gain coincide, as is typical for an LTI system with sinusoidal inputs. With all the pieces in place, the frequency response function~\eqref{eq:def_freq_resp} can be finally specialized in the case of LTI systems.

\begin{proposition} \label{th:LTI-freq}
	Consider the LTI system~\eqref{eq:LTI-system} with transfer function $H : \C \to \C$, driven by the harmonic signal generator~\eqref{eq:LTI-exosystem} with parameter $\omega = (\varpi, a_u, \phi_u) \in \R_+^2 \times \R$ and initial condition $z_0(\omega) = (a_u \sin\phi_u, a_u \cos\phi_u) \in \R^2$.
	Then, the associated $\omega$-radius function~\eqref{eq:def_w-radius} satisfies
	\begin{equation} \label{eq:w-freq-linear}
		\Gamma(\omega) = H(j\omega).
	\end{equation}
\end{proposition}

\begin{proof}
	Building on Proposition~\ref{th:LTI-gain}, Proposition~\ref{th:LTI-phase}, and Proposition~\ref{th:LTI-radius}, we have $\Gamma(\omega)
		=
		\alpha(\omega) r(\omega) e^{j \vartheta(\omega)}
		=
		\norm*{ H(j\omega) }
		e^{j \arg\pr*{H(j\omega)} }
		=
		H(j\omega)$, that is~\eqref{eq:w-freq-linear}.
\end{proof}

\begin{figure*}[t]
	\centering
	\condVSpace{-1.5em}
	\begin{subfigure}{0.325\linewidth}
	\centering
	\includegraphics[width=\linewidth]{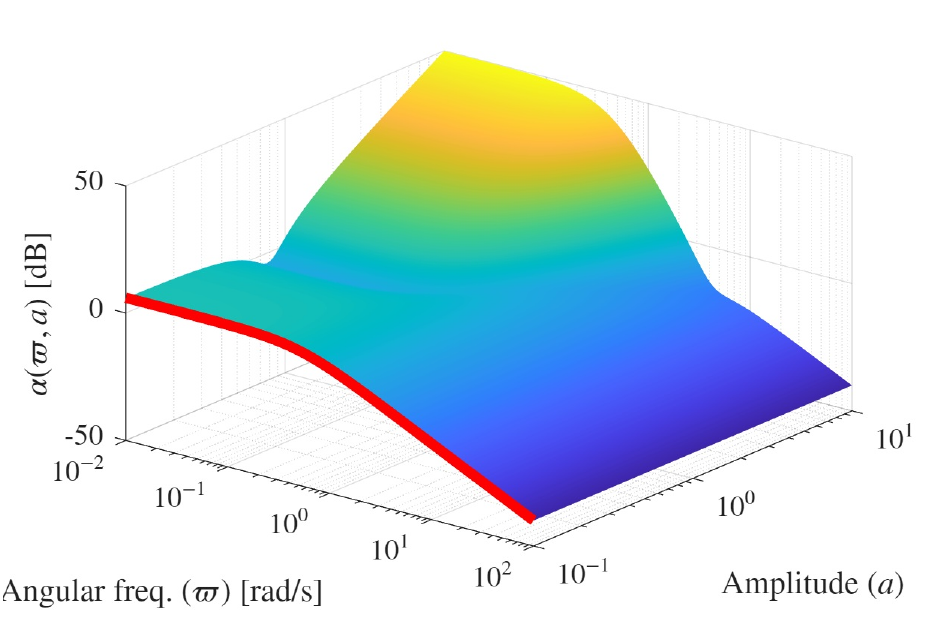}
	\condVSpace{-1.5em}
	\caption{$\omega$-gain}
	\label{fig:bode:ex_base::gain}
	\end{subfigure}
	\hfill
	\begin{subfigure}{0.325\linewidth}
	\centering
	\includegraphics[width=\linewidth]{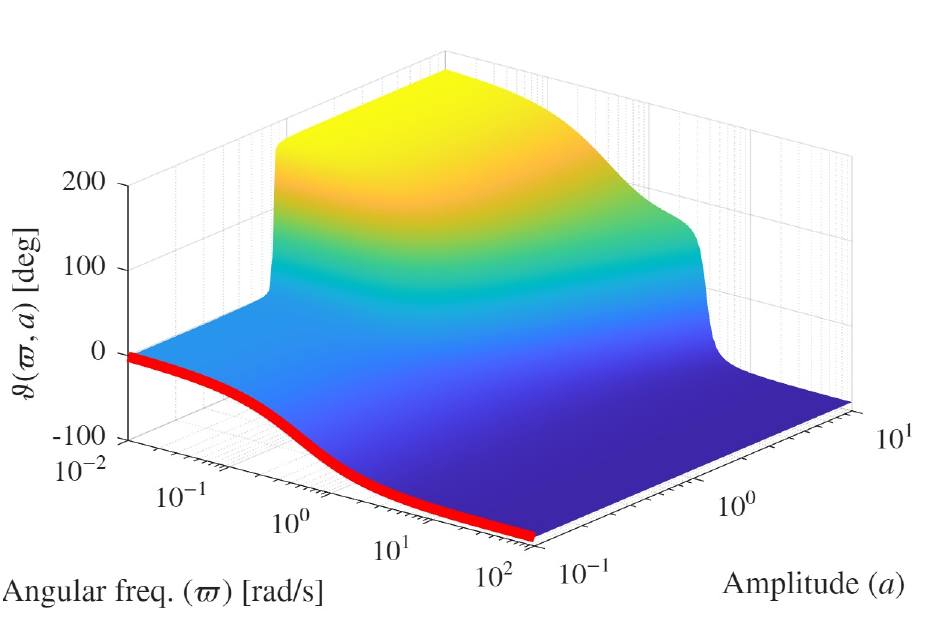}
	\condVSpace{-1.5em}
	\caption{$\omega$-phase}
	\label{fig:bode:ex_base::phase}
	\end{subfigure}
	\hfill
	\begin{subfigure}{0.325\linewidth}
	\centering
	\includegraphics[width=\linewidth]{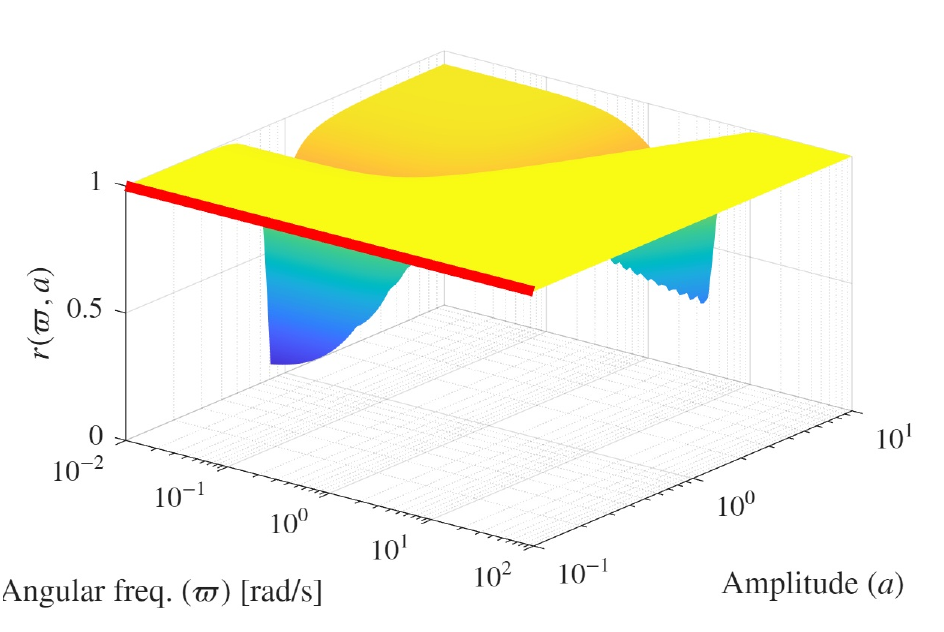}
	\condVSpace{-1.5em}
	\caption{$\omega$-radius}
	\label{fig:bode:ex_base::radius}
	\end{subfigure}
	\caption{
		Bode-like diagram of the nonlinear system~\eqref{eq:nonlinear-system_example_f}-\eqref{eq:nonlinear-system_example_h}, corresponding to the parameter configuration~\eqref{eq:parameters_ex1}, in response to the input~\eqref{eq:LTI-exosystem}, with angular frequency $\varpi \in [10^{-2}, 10^2]$ and amplitude $a_u \in [10^{-2}, 10^2]$.
		The red curve corresponds to the $\omega$-gain, $\omega$-phase, and $\omega$-radius of the linearized system about the equilibrium $x = 0$.
	}
	\label{fig:bode:ex_base}
	\condVSpace{-1.25em}
\end{figure*}

\begin{figure*}[t]
	\centering
	\begin{subfigure}{0.325\linewidth}
	\centering
	\includegraphics[width=\linewidth]{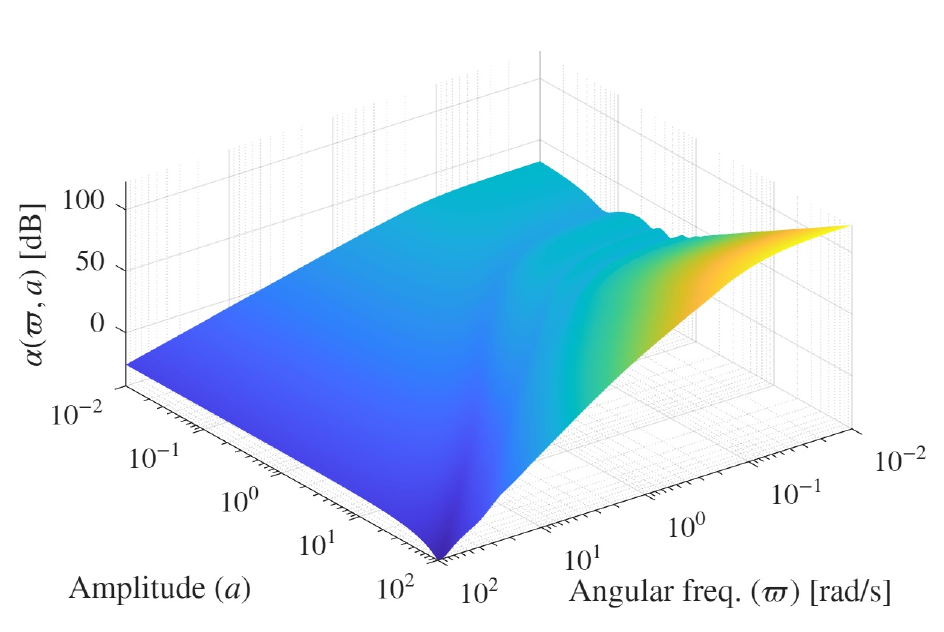}
	\condVSpace{-1.5em}
	\caption{$\omega$-gain}
	\label{fig:bode:ex_cveticanin::gain}
	\end{subfigure}
	\hfill
	\begin{subfigure}{0.325\linewidth}
	\centering
	\includegraphics[width=\linewidth]{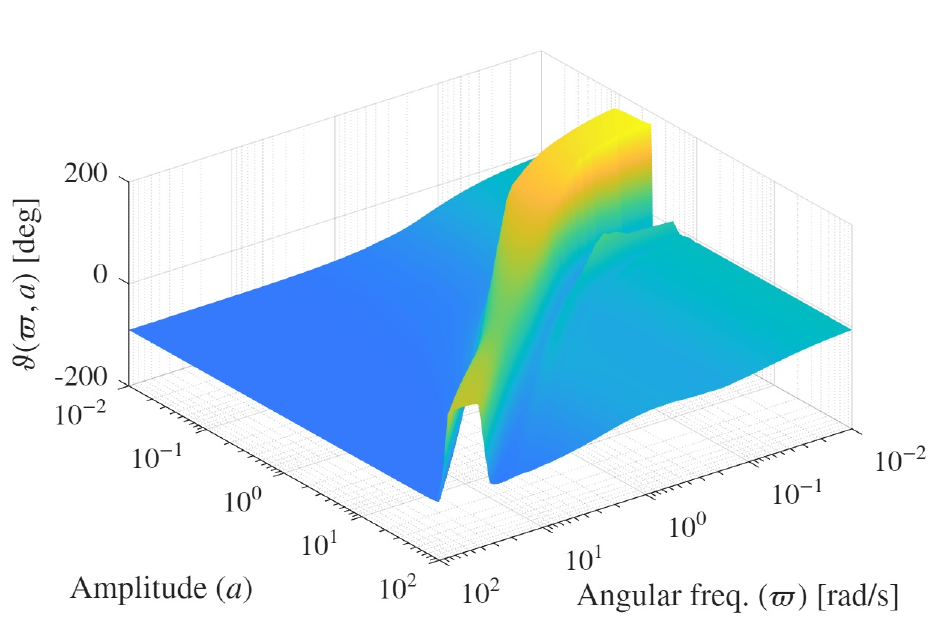}
	\condVSpace{-1.5em}
	\caption{$\omega$-phase}
	\label{fig:bode:ex_cveticanin::phase}
	\end{subfigure}
	\hfill
	\begin{subfigure}{0.325\linewidth}
	\centering
	\includegraphics[width=\linewidth]{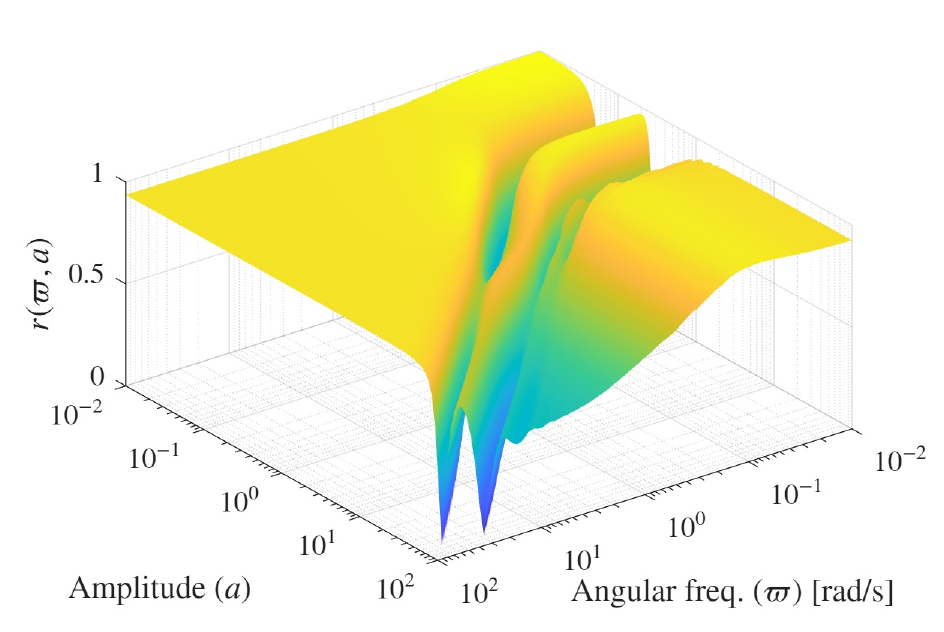}
	\condVSpace{-1.25em}
	\caption{$\omega$-radius}
	\label{fig:bode:ex_cveticanin::radius}
	\end{subfigure}
	\caption{
		Bode-like diagram of the system in Example~\ref{ex:cveticanin_bode}, with angular frequency $\varpi \in [10^{-2}, 10^2]$ and amplitude $a_u \in [10^{-2}, 10^2]$.
	}
	\label{fig:bode:ex_cveticanin}
	\condVSpace{-1.25em}
\end{figure*}

\begin{figure*}[t]
	\centering
	\begin{subfigure}{0.325\linewidth}
	\centering
	\includegraphics[width=\linewidth]{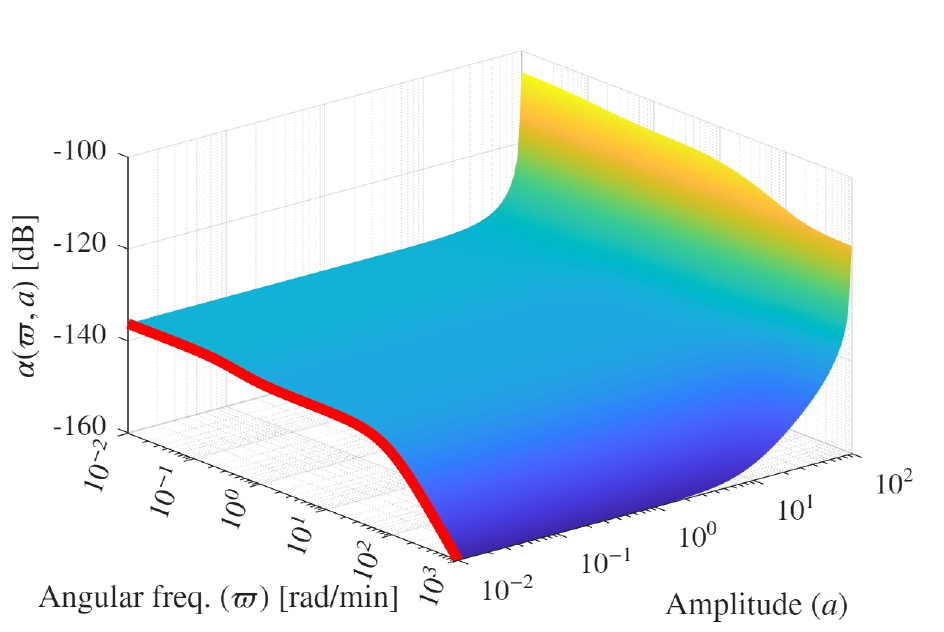}
	\condVSpace{-1.5em}
	\caption{$\omega$-gain}
	\label{fig:bode:ex_bio::gain}
	\end{subfigure}
	\hfill
	\begin{subfigure}{0.325\linewidth}
	\centering
	\includegraphics[width=\linewidth]{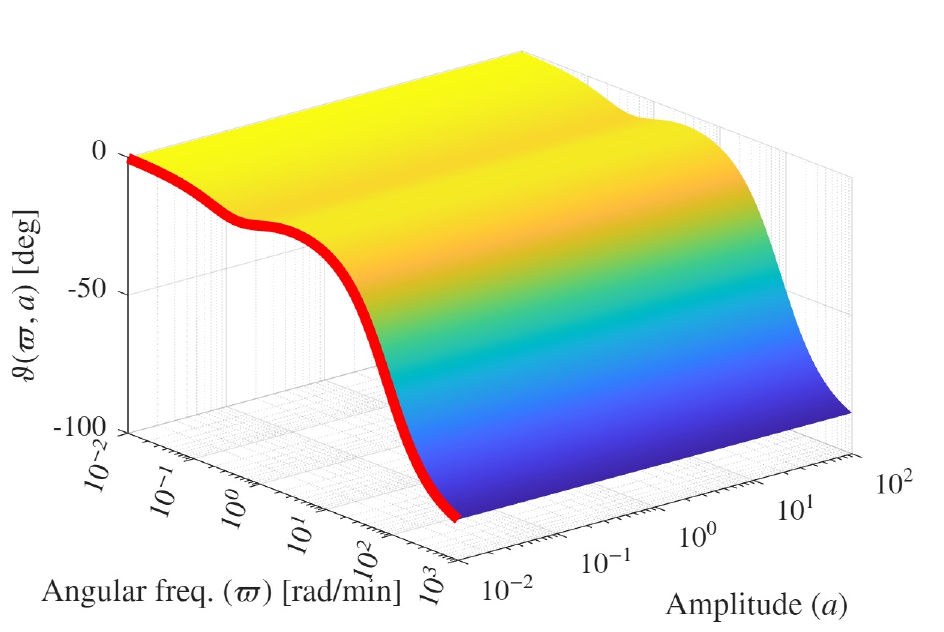}
	\condVSpace{-1.5em}
	\caption{$\omega$-phase}
	\label{fig:bode:ex_bio::phase}
	\end{subfigure}
	\hfill
	\begin{subfigure}{0.325\linewidth}
	\centering
	\includegraphics[width=\linewidth]{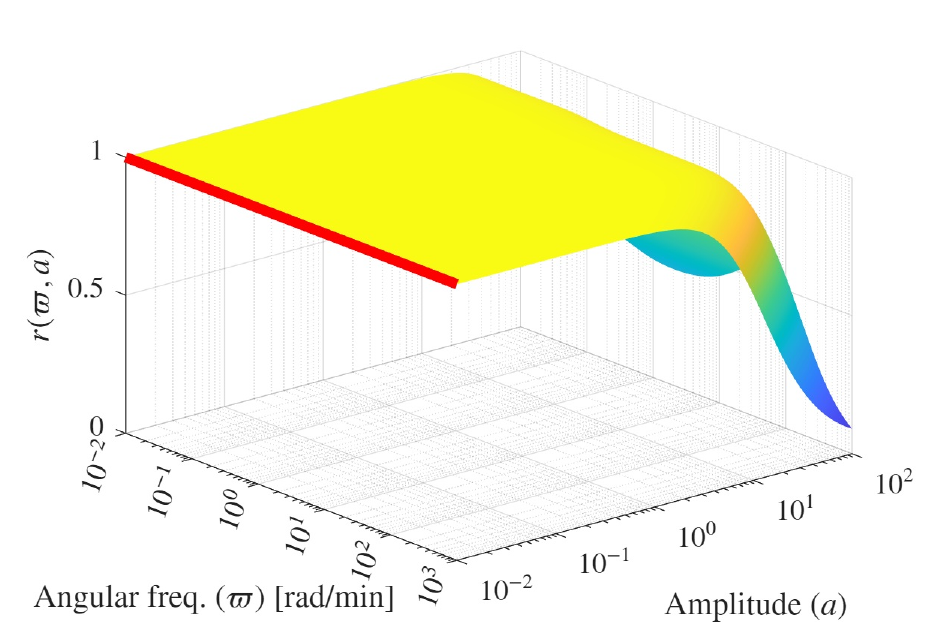}
	\condVSpace{-1.5em}
	\caption{$\omega$-radius}
	\label{fig:bode:ex_bio::radius}
	\end{subfigure}
	\caption{
		Bode-like diagram of the system in Example~\ref{ex:bio_bode}, with angular frequency $\varpi \in [10^{-2}, 10^2]$ and amplitude $a_u \in [10^{-2}, 10^2]$.
		The red curve corresponds to the $\omega$-gain, $\omega$-phase, and $\omega$-radius of the linearized system about the equilibrium $x = 0$.
	}
	\label{fig:bode:ex_bio}
	\condVSpace{-1.5em}
\end{figure*}

\begin{example} \label{ex:base_bode}
	We revisit Example~\ref{ex:base} and analyze the additional information encoded in the associated frequency response function~\eqref{eq:def_freq_resp}.
	In particular, we compute $\alpha(\omega)$, $\vartheta(\omega)$, and $r(\omega)$ from~\eqref{eq:def_w-gain},~\eqref{eq:def_w-phase}, and~\eqref{eq:def_w-radius} respectively, and then perform a Bode-like diagram over input frequency and amplitude.
	Specifically, Figures~\ref{fig:bode:ex_base::gain},~\ref{fig:bode:ex_base::phase} and~\ref{fig:bode:ex_base::radius} present, respectively, the nonlinear Bode-like gain diagram, the nonlinear Bode-like phase diagram and the nonlinear Bode-like radius diagram, derived using the $\omega$-response function under the same parameter configuration specified in equation~\eqref{eq:parameters_ex1}. Figure~\ref{fig:bode:ex_base::gain} shows that for very small input amplitudes, the $\alpha(\omega)$ closely matches the magnitude response of the linearized system.
	As the input amplitude $a_u$ increases and for low-frequency values $\varpi$, near $a_u \approx 0.75$, the $\omega$-gain function $\alpha(\omega)$ exhibits attenuation not present in the linear model.
	For $a_u > 1$, amplification increases with input amplitude.
	As $\omega$ increases, the amplitude response converges to that of the linearized system even at higher amplitudes $a_u$.
	Similarly, Figure~\ref{fig:bode:ex_base::phase} shows that the phase response matches the linearized model at small amplitudes.
	However, near $a_u \approx 0.75$ and low $\omega$, a phase shift of about $-180^\circ$ occurs, consistent with phase inversion and trajectory changes observed in Figure~\ref{fig:ex:base:ts_Y}.
	As $\omega$ increases, the phase again aligns with the linearized response over the same amplitude range.
	Instead, Figure~\ref{fig:bode:ex_base::radius} shows that, at low amplitudes $a_u$ or at large frequencies $\varpi$, the $\omega$-radius is close to $1$, indicating that the nonlinear system does not distort the input.
	In particular, since the input is a sine wave generated by a harmonic oscillator, the $\omega$-radius shows that the output is a sine wave that differs from the input in amplitude and phase (as shown in~\ref{fig:bode:ex_base::gain} and~\ref{fig:bode:ex_base::phase}).
	Instead, around $a_u \approx 0.75$ and low frequencies, the $\omega$-radius significantly decreases, highlighting that the output is a heavily distorted version of the input.
	This is also confirmed by Figure~\ref{fig:ex:base:ts_Y}, where the time series of $\omega$-response is graphically represented.
\end{example}
\begin{example} \label{ex:cveticanin_bode}
	Consider the nonsmooth nonlinear system~\eqref{eq:driven_sys} with state $x = (x_1, x_2) \in \R^2$, and
	\begin{align*}
		f(x, u)
		&
		=
		\pr*{
			\begin{array}{c}
				- 0.1 x_1 + u \\
				- 10 \norm*{\sin(x_1)} - 3 x_2 + 2 x_1
			\end{array}
		},
		\\
		h(x, u)
		&
		= 5 x_1 + 3 x_1 x_2.
	\end{align*}
	We are interested in analyzing its frequency response when fed with the nonlinear, nonsmooth oscillator analyzed in~\cite{Cveticanin2009a}, which can be represented in the form~\eqref{eq:driving_sys} with $\omega = (\varpi, a_u)\in \R^2$, $z_0(\omega) = (a_u, 0) \in \R^2$ and
	\begin{align*}
		s(\omega, z)
		&
		=
		\pr*{
			\begin{array}{c}
				z_2 \\
				-\pr*{\frac{\varpi k_\alpha}{2 \pi}}^2 z_1 \norm*{\frac{z_1}{a_u}}^{\alpha - 1}
			\end{array}
		}, 
		&
		\ell(\omega, z)
		&
		= z_1
	\end{align*}
	where $k_\alpha = \frac{4}{\sqrt{2 (\alpha + 1 )}} B\pr*{ \frac{1}{\alpha+1},\frac{1}{2} }$, $\alpha = 30$, and $B$ is the Euler's beta function~\cite[Sec 8.17]{Olver2010a}.
	Then, the input $u$ has angular frequency equal to $\varpi$, and it oscillates within the interval $[-a_u, a_u]$. Note that neither the system nor the signal generator admits a linearization around the origin as they are both non-differentiable, and hence representations by Bode diagrams are not allowed.
	On the other hand, the $\omega$-gain, $\omega$-phase, and $\omega$-radius are still well-defined quantities that they can be used for system analysis.
	The system performance at different input angular frequency $\varpi$ and amplitude $a_u$ are illustrated in the Bode-like diagram of Figure~\ref{fig:bode:ex_cveticanin}.
\end{example}
\begin{example} \label{ex:bio_bode}
	Consider the gene regulatory model in~\cite[Eq.~6.3]{del2015biomolecular}, in which a protein $x_1$ binds reversibly to a downstream promoter $p$, forming the complex $x_2$ with association and dissociation rates $k_{\mathrm{on}}>0$ and $k_{\mathrm{off}}>0$, and total promoter concentration $p_{\mathrm{tot}} = p + x_2$.
	Here, we additionally assume saturation~\cite[Ch.~1]{keener2025mathematical} in both production and degradation of the protein $x_1$.
	The transcription factor input is given by $u(t) = \bar{u} + \tilde{u}(t)$, with $\tilde{u}(t) = a_u \sin(\varpi t)$ and $\bar v > a_u > 0$, and the dynamics of $x_1$ and $x_2$ with saturating production and degradation yields
	\condVSpace{-0.25em}
	\begin{align*}
		\dot{x}_1 & = w(x_1,x_2) - \frac{\gamma x_1}{K + x_1} + \frac{\bar{u} + \tilde{u}}{K_{\mathrm{u}} + \bar{u} + \tilde{u}},\\[-0.5mm]
		\dot{x}_2 & = - w(x_1,x_2), \\
		y & = x_1,
	\end{align*}
	\condVSpace{-0.25em}
	with retroactive effect of the promoter $w(x_1,x_2) = k_{\mathrm{off}} x_2 - k_{\mathrm{on}} (p_{\mathrm{tot}} - x_2)x_1$, decay rate $\gamma>0$, and saturation weights $K>0$ and $K_{\mathrm{u}}>0$.
	For zero frequency ($\varpi = 0$), the system has a globally asymptotically stable equilibrium given by $x_1^\ast = \frac{K \bar{u}}{\gamma (K_{\mathrm{u}} + \bar{u}) - \bar{u}}$ and $x_2^\ast = \frac{k_{\mathrm{on}} p_{\mathrm{tot}} x_1^\ast}{k_{\mathrm{off}} + k_{\mathrm{on}} x_1^\ast}$.
	However, to perform frequency response analysis we consider a shifted version of the system, so that the equilibrium lies at the origin and the frequency response function can be computed from input $\tilde{u}(t) = u(t) - \bar{u}$ and output $\tilde{y}(t) = x_1(t) - x_1^\ast$.
	\diff{In this example, we consider the case with $\gamma=5$, $k_{\mathrm{on}} = 1$, $k_{\mathrm{off}} = 0.5$, $p_{\mathrm{tot}}=20$, $\bar{u}=100.5$, $K = 0.05$ and $K_{\mathrm{u}}=0.1$.}

	Bode-like diagrams illustrating the performance at various angular frequency $\varpi$ and amplitude $a_u$ are reported in Figure~\ref{fig:bode:ex_bio}.
	Specifically, Figure~\ref{fig:bode:ex_bio::gain} shows that the $\omega$-gain is consistent with that of the linearized system until the effect of saturation becomes visible and dominates the linear effect.
	This is consistent with the shape of the $\omega$-radius in Figure~\ref{fig:bode:ex_bio::radius}, which proves that, at high amplitudes, the system bends the shape of the input, achieving high nonlinearity at high frequencies.
	Finally, Figure~\ref{fig:bode:ex_bio::phase} confirms that the phase remains unchanged throughout the input amplitudes.
\end{example}

\section{Analysis of Frequency Response of Nonlinear Systems}

\label{sec:analysis_of_frf}

\subsection{Input-Output Properties of the Frequency Response}

One of the most widely studied input-output properties of nonlinear systems is that of \emph{dissipativity}.
The foundational theory of dissipativity was introduced by Willems~\cite{willems1972dissipative}.
Since then, it has been recognized as a unifying framework for characterizing both input-output and internal qualitative behaviors of dynamical systems~\cite{hill1980dissipative}.
For the purpose of this section, we consider dissipativity in an input-output sense and adopt a weaker notion that does not require the existence of a storage function.
We say that the system~\eqref{eq:driven_sys} is \emph{dissipative}, with continuous function $s : \R^m \times \R^m \to \R$ referred to as the supply rate%
\footnote{A supply rate is such that $s(u,y)$ is a locally Lebesgue integrable function for all admissible input $u$ and initial state $x(0)$.},
if for all admissible inputs $u \in \mathcal{L}_2(\R, \R^m)$, the corresponding output $y \in \mathcal{L}_2(\R, \R^m)$ satisfies the condition
\begin{equation} \label{eq:dissipation_inequality}
	\forall u \in \mathcal{L}_2(\R, \R^m),
	\hspace{1ex}
	\liminf_{\tau\to\infty} \int_0^\tau s(u(t), y(t)) \, dt > -\infty.
\end{equation}
This condition ensures that the system does not dissipate an unbounded amount of energy over an infinite time horizon.
However, on its own, it does not assume or imply the existence of a storage function that quantifies stored energy, see~\cite{brogliato2020dissipative}.

In addition to considering dissipativity with respect to all inputs, we introduce an input-dependent notion of dissipativity tailored to the class of inputs generated by~\eqref{eq:driving_sys}. 
\begin{definition}
	The system~\eqref{eq:driven_sys} is said \emph{dissipative relative to~\eqref{eq:driving_sys}} if, for all $\omega \in \mathcal{N}$ and inputs of the form $u=\ell(\omega, z)$, the corresponding output $y \in \mathcal{L}_2(\R, \R^m)$ satisfies the condition
	\begin{equation} \label{eq:dissipativity_wrt_input}
		\forall \omega \in \mathcal{N},
		\hspace{3ex}
		\liminf_{\tau\to\infty} \int_0^\tau s(\ell(\omega, z(t)), y(t)) \, dt > -\infty.
	\end{equation}
\end{definition}
This notion generalizes the concept of dissipativity with respect to arbitrary inputs, in the sense that~\eqref{eq:dissipation_inequality} implies the input-dependent condition~\eqref{eq:dissipativity_wrt_input}, whereas the converse does not necessarily hold.

In what follows, we demonstrate how dissipativity implies properties of the $\omega$-response.
\begin{lemma}\label{prop:beta_implies_zero}
	Let the system~\eqref{eq:driven_sys} be dissipative relative to~\eqref{eq:driving_sys} with supply rate $s$, and suppose that the invariant submanifold $\mathcal{M}$ in~\eqref{eq:inv_manifold} for the interconnected system~\eqref{eq:driven_sys}-\eqref{eq:driving_sys} is attractive.
	Then, the $\omega$-response of the system~\eqref{eq:driven_sys} to the input~\eqref{eq:driving_sys} satisfies
	\begin{equation}\label{eq:dissipativity_on_manifold}
		\forall \omega \in \mathcal{N},
		\hspace{3ex}
		\int_{\T_\omega} s(\ell(\omega, z(t)), Y(\omega, z(t))) \, dt \ge 0.
	\end{equation}
\end{lemma}

\begin{proof}
	By assumption, the system is dissipative relative to~\eqref{eq:driving_sys}, that is~\eqref{eq:dissipativity_wrt_input} holds.
	Note that, over the interval $[0, \tau]$, the integral in~\eqref{eq:dissipation_inequality} can be equivalently written as
	\begin{equation*}
		\int_0^\tau s(\ell(\omega, z(t)), y(t)) \, dt
		=
		\sum_{k=0}^{\bar{\tau}} \Delta_k,
	\end{equation*}
	where $\tau = \bar{\tau} T_\omega + \tau_0$, $\bar{\tau}$ is the unique integer such that $\tau_0 \in [0, T_\omega)$, and the sequence $\Delta_k : \N \to \R$ is defined by
	\begin{equation}\label{eq:sequence_delta}
		\begin{aligned}
			\Delta_0
			&
			\coloneq \int_0^{\tau_0} s(\ell(\omega, z(t)), y(t))\, dt,
			\\
			\Delta_k
			&
			\coloneq \int_{\tau_0 + (k-1) T_\omega}^{\tau_0 + k T_\omega} s(\ell(\omega, z(t)), y(t))\, dt, \quad \forall k \ge 1 .
		\end{aligned}
	\end{equation}
	Continuity of the functions $s$ and $\ell$, as well as of the trajectories $z$ and $y$, ensures boundedness of $\Delta_k$, that is there exist constants $\underline{q}$ and $\overline{q}$ such that $\underline{q} \le \Delta_k \le \overline{q}$ for all $k$.
	Furthermore, by the assumption that the submanifold $\mathcal{M}$ is invariant and attractive (along the trajectories of the interconnected system~\eqref{eq:driven_sys}-\eqref{eq:driving_sys}), we have that the sequence $\Delta_k$ converges to the constant
	\begin{equation}\label{eq:c_omega}
		c_{\omega} \coloneq \int_{\T_\omega} s(\ell(\omega, z(t)), Y(\omega, z(t))) \, dt, \hspace{3ex} \forall \omega \in \mathcal{N}.
	\end{equation}
	To prove the statement, we need to show that $ c_{\omega} \ge 0$ under the given assumptions.
	Therefore, \emph{ex absurdo}, consider the case in which $c_{\omega} < 0$.
	Then, since the sequence $\Delta_{k}$ converges to $c_{\omega}$, it follows that
	\begin{equation*}
		\lim_{k\to\infty} \Delta_{k} = c_{\omega} < 0.
	\end{equation*}
	This implies the existence of a constant $b>0$ and an index $K \in \N$ such that for all $k\ge K$ we have $\Delta_{k}<-b$.
	Consequently, the infinite sum of the series diverges to negative infinity, \emph{i.e.},
	\begin{equation*}
		\lim_{\tau \to +\infty} \sum_{k = 0}^{\tau} \Delta_{k} = -\infty,
	\end{equation*}
	which contradicts the lower boundedness of the integral~\eqref{eq:dissipativity_wrt_input}.
	On the contrary, for $c_{\omega} \ge 0$, we infer that
	\begin{equation*}
		\lim_{k\to\infty} \Delta_{k} = \int_{\T_\omega} s(\ell(\omega, z(t)), Y(\omega, z(t))) \, dt \ge 0, 
	\end{equation*}
	which implies that the infinite sum is lower bounded, and thus proves the claim.
\end{proof}

\begin{lemma}\label{prop:zero_implies_beta}
	Let the system~\eqref{eq:driven_sys} be such that its $\omega$-response to the input~\eqref{eq:driving_sys} satisfies~\eqref{eq:dissipativity_on_manifold}.
	Suppose that the invariant submanifold $\mathcal{M}$ in~\eqref{eq:inv_manifold} is attractive.
	Then, the system~\eqref{eq:driven_sys} is dissipative relative to~\eqref{eq:driving_sys} with supply rate $s$.
\end{lemma}

\begin{proof}
	Let us recall the sequence $\Delta$ and its limit $c_\omega \in \R$ as defined in~\eqref{eq:sequence_delta} and~\eqref{eq:c_omega}.
	By assumption that the $\omega$-response to the input~\eqref{eq:driving_sys} satisfies~\eqref{eq:dissipativity_on_manifold}, we have that $c_{\omega} \ge 0$.
	Moreover, following the boundedness of the sequence $\Delta_k$, we have that, if $c_\omega > 0$,
	\begin{equation*}
		\lim_{\tau \to \infty} \int_0^\tau s(\ell(\omega, z(t)), y(t)) \, dt
		=
		\lim_{\tau \to \infty} \sum_{k = 0}^{\tau} \Delta_{k} = \infty,
	\end{equation*}
	and, if $c_\omega = 0$,
	\begin{align*}
		\lim_{\tau \to \infty} \int_0^\tau s(\ell(\omega, z(t)), y(t)) \, dt
		=
		\lim_{\tau \to \infty} \sum_{k = 0}^{\tau} \Delta_{k} < \infty.
	\end{align*}
	Hence, since both the above conditions imply the existence of a lower bound for the cumulative sum, we infer that the inequality~\eqref{eq:dissipativity_wrt_input} holds, thereby proving the claim.
\end{proof}
Since Lemma~\ref{prop:beta_implies_zero} establishes that
\diff{$\eqref{eq:dissipativity_wrt_input} \Rightarrow~\eqref{eq:dissipativity_on_manifold}$},
and Lemma~\ref{prop:zero_implies_beta} establishes that
$\eqref{eq:dissipativity_wrt_input} \Leftarrow~\eqref{eq:dissipativity_on_manifold}$,
we infer that $\eqref{eq:dissipativity_wrt_input} \Leftrightarrow~\eqref{eq:dissipativity_on_manifold}$,
which yields a necessary and sufficient condition for dissipativity relative to~\eqref{eq:driving_sys}, characterized in terms of the $\omega$-response.
\begin{theorem} \label{th:iff}
	Consider the system~\eqref{eq:driven_sys} and the input~\eqref{eq:driving_sys}.
	Suppose that the invariant submanifold $\mathcal{M}$ in~\eqref{eq:inv_manifold} is attractive.
	Then, the system~\eqref{eq:driven_sys} is dissipative relative to~\eqref{eq:driving_sys} with supply rate $s$ if and only if the $\omega$-response of the system~\eqref{eq:driven_sys} to the input~\eqref{eq:driving_sys} satisfies~\eqref{eq:dissipativity_on_manifold}.
\end{theorem}
This characterization of dissipativity in terms of the $\omega$-response forms the foundation for analyzing bounds on the nonlinear frequency response~\eqref{eq:def_freq_resp} for various supply rate $s$.

\begin{corollary}[$\mathcal{L}_2$ stability] \label{th:L2_gain}
	The system~\eqref{eq:driven_sys} is $\mathcal{L}_2$ stable relative to~\eqref{eq:driving_sys} with $\gamma > 0$, \emph{i.e.},~\eqref{eq:dissipativity_wrt_input} holds with
	\begin{equation} \label{eq:def_L2_gain}
		s(u, y) = \gamma^2 u^\top u - y^\top y,
	\end{equation}
	and $u = \ell(\omega, z)$, if and only if, for all $\omega \in \mathcal{N}$, $\alpha(\omega) \in \ps{0, \gamma}$.
\end{corollary}
\begin{proof}
	By substituting~\eqref{eq:def_L2_gain} into~\eqref{eq:dissipativity_wrt_input} and applying Theorem~\ref{th:iff}, we obtain the following condition in terms of the $\omega$-response
	$
	0 \le \gamma^2 \norm[\big]{\ell(\omega, z)}_{\mathcal{L}_2(\T_\omega)}^2 - \norm[\big]{Y(\omega, z)}_{\mathcal{L}_2(\T_\omega)}^2,
	$
	which implies
	$
	\norm[\big]{Y(\omega, z)}_{\mathcal{L}_2(\T_\omega)} \le \gamma \norm[\big]{\ell(\omega, z)}_{\mathcal{L}_2(\T_\omega)}.
	$
	Since $\norm{\ell(\omega, z)}_{\mathcal{L}_2(\T_\omega)} \ne 0$, dividing both sides by $\norm{\ell(\omega, z)}_{\mathcal{L}_2(\T_\omega)}$ and using~\eqref{eq:def_w-gain}, we conclude that the $\omega$-gain is bounded by $\gamma$.
\end{proof}

\begin{corollary}[Passivity] \label{th:passivity}
	The system~\eqref{eq:driven_sys} is passive relative to~\eqref{eq:driving_sys}, \emph{i.e.},~\eqref{eq:dissipativity_wrt_input} holds with
	\begin{equation}\label{eq:def_passivity}
		s(u, y) = u^\top y,
	\end{equation}
	and $u = \ell(\omega, z)$, if and only if for all $\omega \in \mathcal{N}$, $\alpha(\omega) = 0$, or
	\begin{equation*}
		\vartheta(\omega) \in \ps*{-\frac{\pi}{2}, \ \frac{\pi}{2}}.
	\end{equation*}
\end{corollary}
\begin{proof}
	Applying~\eqref{eq:def_passivity} in~\eqref{eq:dissipativity_wrt_input} and invoking Theorem~\ref{th:iff}, we obtain the necessary and sufficient condition
	\begin{equation} \label{th:passivity::cond}
		\inner[\big]{\ell(\omega,z)}{Y(\omega,z)}_{\mathcal{L}_2(\T_\omega)} \ge 0.
	\end{equation}
	If $\alpha(\omega) = 0$, then $Y(\omega,z) = 0$ and~\eqref{th:passivity::cond} holds.
	In all the other cases, $\norm*{Y(\omega,z)} > 0$ and the $\omega$-phase is well-defined.
	Recalling~\eqref{eq:cos_sin} and the definition of $\mathfrak{Re}(\omega)$,~\eqref{th:passivity::cond} holds if and only if $\cos(\vartheta(\omega)) \ge 0$.
	Then, we conclude the proof by noting that the cosine is non-negative if and only if $\vartheta(\omega) \in \ps{-\frac{\pi}{2}, \ \frac{\pi}{2}}$.
\end{proof}

\begin{corollary}[Counterclockwise dynamics] \label{th:counterclockwise}
	The system~\eqref{eq:driven_sys} is with counterclockwise input-output dynamics relative to~\eqref{eq:driving_sys}, \emph{i.e.},~\eqref{eq:dissipativity_wrt_input} holds with
	\begin{equation} \label{eq:def_counterclockwise}
		s(u, y) = u^\top \dot{y},
	\end{equation}
	and $u = \ell(\omega, z)$, if and only if for all $\omega \in \mathcal{N}$, $\alpha(\omega) = 0$, or
	\begin{equation*}
		\vartheta(\omega)\in \ps*{-\pi, \ 0}.
	\end{equation*}
\end{corollary}
\begin{proof}
	Substituting~\eqref{eq:def_counterclockwise} into~\eqref{eq:dissipativity_wrt_input} and applying Theorem~\ref{th:iff}, we obtain the following necessary and sufficient condition $\inner[\big]{\ell(\omega,z)}{\dot{Y}(\omega,z)}_{\mathcal{L}_2(\T_\omega)} \ge 0$.
	Since $\ell(\omega,z(t))$ and $Y(\omega,z(t))$ are $T_\omega$-periodic, integration by parts over $\T_\omega$ yields
	\begin{align*}
		\inner[\big]{\ell(\omega,z)}{\dot{Y}(\omega,z)}_{\mathcal{L}_2(\T_\omega)}
		=
		&
		- \inner[\big]{\dot{\ell}(\omega,z)}{Y(\omega,z)}_{\mathcal{L}_2(\T_\omega)}
		\\
		&
		+ \ell(\omega,z(T_\omega))^\top Y(\omega,z(T_\omega)) \\
		&
		- \ell(\omega,z(0))^\top Y(\omega,z(0)) .
	\end{align*}
	Due to periodicity, the boundary terms cancel, and thus
	\begin{equation*}
		\inner[\big]{\ell(\omega,z)}{\dot{Y}(\omega,z)}_{\mathcal{L}_2(\T_\omega)}
		=
		- \inner[\big]{\dot{\ell}(\omega,z)}{Y(\omega,z)}_{\mathcal{L}_2(\T_\omega)},
	\end{equation*}
	which implies
	\begin{equation} \label{th:counterclockwise::cond}
		\inner[\big]{\dot{\ell}(\omega,z)}{Y(\omega,z)}_{\mathcal{L}_2(\T_\omega)} \le 0.
	\end{equation}
	If $\alpha(\omega) = 0$, then $Y(\omega,z) = 0$ and~\eqref{th:passivity::cond} holds.
	In all the other cases, $\norm*{Y(\omega,z)} > 0$ and the $\omega$-phase is well-defined.
	Recalling~\eqref{eq:cos_sin} and the definition of $\mathfrak{Im}(\omega)$,~\eqref{th:counterclockwise::cond} holds if and only if $\sin(\vartheta(\omega)) \le 0$.
	Then, we conclude the proof by noting that the sine is non-positive if and only if $\vartheta(\omega) \in \ps{-\pi, 0}$.
\end{proof}

\begin{corollary}[Output strict passivity] \label{th:output_strictly_passivity}
	The system~\eqref{eq:driven_sys} is output strictly passive relative to~\eqref{eq:driving_sys} with $\gamma_1 > 0$, \emph{i.e.},~\eqref{eq:dissipativity_wrt_input} holds with
	\begin{equation} \label{eq:def_output_strictly_passivity}
		s(u, y) = u^\top y - \gamma_1 y^\top y,
	\end{equation}
	and $u = \ell(\omega, z)$, if and only if, for all $\omega \in \mathcal{N}$, $\alpha(\omega) = 0 $ or
	\begin{equation*}
		\alpha(\omega) \in
		\left( 0, \frac{r(\omega)\cos(\vartheta(\omega))}{\gamma_1} \right]
		\hspace{1ex}
		\land
		\hspace{1ex}
		\vartheta(\omega)\in \pr*{ -\frac{\pi}{2}, \frac{\pi}{2} }.
	\end{equation*}
\end{corollary}

\begin{proof}
	Substituting~\eqref{eq:def_output_strictly_passivity} into~\eqref{eq:dissipativity_wrt_input} and using Theorem~\ref{th:iff} yields
	\begin{equation} \label{th:output_strictly_passivity::cond}
		\inner{\ell(\omega,z)}{Y(\omega,z)}_{\mathcal{L}_2(\T_\omega)}
		\ge
		\gamma_1 \norm{Y(\omega,z)}_{\mathcal{L}_2(\T_\omega)}^2.
	\end{equation}
	If $\alpha(\omega) = 0$, then $Y(\omega,z) = 0$ and the inequality~\eqref{th:output_strictly_passivity::cond} becomes $\inner{\ell(\omega,z)}{0}_{\mathcal{L}_2(\T_\omega)} \ge 0$ which is true.
	In all the other cases, $\norm*{Y(\omega,z)} > 0$ and the $\omega$-phase is well-defined.
	Therefore, we can divide both side of~\eqref{th:output_strictly_passivity::cond} by $\norm{\ell(\omega,z)}_{\mathcal{L}_2(\T_\omega)}\norm{Y(\omega,z)}_{\mathcal{L}_2(\T_\omega)}$, obtaining
	\begin{align}
		\gamma_1
		\frac
		{\norm{Y(\omega,z)}_{\mathcal{L}_2(\T_\omega)}}
		{\norm{\ell(\omega,z)}_{\mathcal{L}_2(\T_\omega)}}
		&
		\le
		\frac{\inner{\ell(\omega,z)}{Y(\omega,z)}_{\mathcal{L}_2(\T_\omega)}}{\norm{\ell(\omega,z)}_{\mathcal{L}_2(\T_\omega)}\norm{Y(\omega,z)}_{\mathcal{L}_2(\T_\omega)}}
		\nonumber
		\\
		&
		=
		\mathfrak{Re}(\omega)
		=
		r(\omega) \cos(\vartheta(\omega)),
		\label{th:output_strictly_passivity::cond2}
	\end{align}
	where we used the definition of $\mathfrak{Re}(\omega)$ and~\eqref{eq:cos_sin}.
	Since $\gamma_1$ and $r(\omega)$ are positive, we note that~\eqref{th:output_strictly_passivity::cond2} holds only if $\cos(\vartheta(\omega)) > 0$, and thus $\vartheta(\omega) = \pr*{ -\frac{\pi}{2}, \frac{\pi}{2} }$, see Corollary~\ref{th:passivity}.
	Finally, recalling the definition of $\alpha(\omega)$ and that $\gamma_1 > 0$, we end the proof by noting~\eqref{th:output_strictly_passivity::cond2} holds if and only if
	\begin{equation*}
		0
		<
		\alpha(\omega)
		=
		\frac{\norm{Y(\omega,z)}_{\mathcal{L}_2(\T_\omega)}}{\norm{\ell(\omega,z)}_{\mathcal{L}_2(\T_\omega)}}
		\le
		\frac{r(\omega)\cos(\vartheta(\omega))}{\gamma_1}.
	\end{equation*}
\end{proof}

\begin{corollary}[Input strict passivity] \label{th:input_strictly_passivity}
	The system~\eqref{eq:driven_sys} is input strictly passive relative to~\eqref{eq:driving_sys} with $\gamma_2 > 0$, \emph{i.e.},~\eqref{eq:dissipativity_wrt_input} holds with
	\begin{equation} \label{eq:def_input_strictly_passivity}
		s(u, y) = u^\top y - \gamma_2 u^\top u,
	\end{equation}
	and $u = \ell(\omega, z)$, if and only if, for all $\omega \in \mathcal{N}$,
	\begin{equation*}
		\alpha(\omega)
		\in
		\left[\frac{\gamma_2}{r(\omega)\cos(\vartheta(\omega))}, \ \infty \right)
		\hspace{1ex}
		\land
		\hspace{1ex}
		\vartheta(\omega) \in \pr*{-\frac{\pi}{2}, \frac{\pi}{2}}.
	\end{equation*}
\end{corollary}

\begin{proof}
	Substituting~\eqref{eq:def_input_strictly_passivity} into~\eqref{eq:dissipativity_wrt_input} and using Theorem~\ref{th:iff} yields
	\begin{equation} \label{th:input_strictly_passivity::cond}
		\inner{\ell(\omega,z)}{Y(\omega,z)}_{\mathcal{L}_2(\T_\omega)}
		\ge
		\gamma_2 \norm{\ell(\omega,z)}_{\mathcal{L}_2(\T_\omega)}^2.
	\end{equation}
	If $\alpha(\omega) = 0$, then $Y(\omega,z) = 0$ and the inequality~\eqref{th:output_strictly_passivity::cond} becomes $\inner{\ell(\omega,z)}{0}_{\mathcal{L}_2(\T_\omega)} \ge \gamma_2 \norm{\ell(\omega,z)}_{\mathcal{L}_2(\T_\omega)}^2 > 0$ which does not hold, because, by assumption, $\norm{\ell(\omega,z)}_{\mathcal{L}_2(\T_\omega)} > 0$.
	In all the other cases, $\norm*{Y(\omega,z)} > 0$ and the $\omega$-phase is well-defined.
	Therefore, we can divide both side of~\eqref{th:input_strictly_passivity::cond} by $\norm{\ell(\omega,z)}_{\mathcal{L}_2(\T_\omega)}\norm{Y(\omega,z)}_{\mathcal{L}_2(\T_\omega)}$, obtaining
	\begin{align}
		\gamma_2
		\frac
		{\norm{\ell(\omega,z)}_{\mathcal{L}_2(\T_\omega)}}
		{\norm{Y(\omega,z)}_{\mathcal{L}_2(\T_\omega)}}
		&
		\le
		\frac{\inner{\ell(\omega,z)}{Y(\omega,z)}_{\mathcal{L}_2(\T_\omega)}}{\norm{\ell(\omega,z)}_{\mathcal{L}_2(\T_\omega)}\norm{Y(\omega,z)}_{\mathcal{L}_2(\T_\omega)}}
		\nonumber
		\\
		&
		=
		\mathfrak{Re}(\omega)
		=
		r(\omega) \cos(\vartheta(\omega)),
		\label{th:input_strictly_passivity::cond2}
	\end{align}
	where we used the definition of $\mathfrak{Re}(\omega)$ and~\eqref{eq:cos_sin}.
	Since $\gamma_2$ and $r(\omega)$ are positive, we note that~\eqref{th:input_strictly_passivity::cond2} holds only if $\cos(\vartheta(\omega)) > 0$, and thus $\vartheta(\omega) = \pr*{ -\frac{\pi}{2}, \frac{\pi}{2} }$, see Corollary~\ref{th:passivity}.
	Then, recalling the definition of $\alpha(\omega)$ and that $\gamma_2 > 0$, we end the proof by noting that~\eqref{th:input_strictly_passivity::cond2} holds if and only if
	\begin{align*}
		0
		<
		\frac
		{\norm{\ell(\omega,z)}_{\mathcal{L}_2(\T_\omega)}}
		{\norm{Y(\omega,z)}_{\mathcal{L}_2(\T_\omega)}}
		&
		\le
		\frac{r(\omega)\cos(\vartheta(\omega))}{\gamma_2}
		\\
		&
		\implies
		\alpha(\omega) \in \left[\frac{\gamma_2}{r(\omega)\cos(\vartheta(\omega))}, \infty \right). 
	\end{align*}
\end{proof}

\begin{corollary}[Very strict passivity]\label{th:very_strictly_passivity}
	The system~\eqref{eq:driven_sys} is very strictly passive relative to~\eqref{eq:driving_sys} with $\gamma_2 > 0$ and $\gamma_1 > 0$, \emph{i.e.},~\eqref{eq:dissipativity_wrt_input} holds with
	\begin{equation}\label{eq:def_very_strictly_passivity}
		s(u, y)
		=
		u^\top y - \gamma_1 y^\top y - \gamma_2 u^\top u,
	\end{equation}
	and $u = \ell(\omega, z)$, if and only if for all $\omega \in \mathcal{N}$,
	\begin{align*}
		\alpha(\omega)
		&
		\in 
		\ps*{
			\frac{r(\omega)\cos(\omega)}{2\gamma_1} - \frac{\sqrt{\delta(\omega)}}{2\gamma_1},
			\frac{r(\omega)\cos(\omega)}{2\gamma_1} + \frac{\sqrt{\delta(\omega)}}{2\gamma_1}
		},
		\\
		\vartheta(\omega)
		&
		\in
		\pr*{-\frac{\pi}{2}, \frac{\pi}{2}},
	\end{align*}
	where $\delta(\omega) \coloneq r(\omega)^2\cos(\omega)^2 - 4\gamma_2\gamma_1 \ge 0$.
\end{corollary}

\begin{proof}
	Substituting~\eqref{eq:def_very_strictly_passivity} into~\eqref{eq:dissipativity_wrt_input} and using Theorem~\ref{th:iff} yields
	\begin{align}
		\inner{\ell(\omega,z)}{Y(\omega,z)}_{\mathcal{L}_2(\T_\omega)}
		\ge
		&
		\nonumber
		\\
		&
		\hspace{-10ex}
		\gamma_1 \norm{Y(\omega,z)}_{\mathcal{L}_2(\T_\omega)}^2
		+
		\gamma_2 \norm{\ell(\omega,z)}_{\mathcal{L}_2(\T_\omega)}^2.
		\label{th:very_strictly_passivity::cond}
	\end{align}
	If $\alpha(\omega) = 0$, then $Y(\omega,z) = 0$ and the inequality~\eqref{th:output_strictly_passivity::cond} becomes $\inner{\ell(\omega,z)}{0}_{\mathcal{L}_2(\T_\omega)} \ge \gamma_2 \norm{\ell(\omega,z)}_{\mathcal{L}_2(\T_\omega)}^2 > 0$ which does not hold, because, by assumption, $\norm{\ell(\omega,z)}_{\mathcal{L}_2(\T_\omega)} > 0$.
	In all the other cases, $\norm*{Y(\omega,z)} > 0$ and the $\omega$-phase is well-defined.
	Therefore, we can divide both side of~\eqref{th:input_strictly_passivity::cond} by $\norm{\ell(\omega,z)}_{\mathcal{L}_2(\T_\omega)}\norm{Y(\omega,z)}_{\mathcal{L}_2(\T_\omega)}$, obtaining
	\begin{align}
		\gamma_1
		\frac
		{\norm{Y(\omega,z)}_{\mathcal{L}_2(\T_\omega)}}
		{\norm{\ell(\omega,z)}_{\mathcal{L}_2(\T_\omega)}}
		+
		\gamma_2
		\frac
		{\norm{\ell(\omega,z)}_{\mathcal{L}_2(\T_\omega)}}
		{\norm{Y(\omega,z)}_{\mathcal{L}_2(\T_\omega)}}
		&
		\le
		r(\omega) \cos(\vartheta(\omega)),
		\label{th:very_strictly_passivity::cond2}
	\end{align}
	where we used the definition of $\mathfrak{Re}(\omega)$ and~\eqref{eq:cos_sin}.
	Since $\gamma_1$, $\gamma_2$ and $r(\omega)$ are positive, we note that~\eqref{th:very_strictly_passivity::cond2} holds only if $\cos(\vartheta(\omega)) > 0$, and thus $\vartheta(\omega) = \pr*{ -\frac{\pi}{2}, \frac{\pi}{2} }$, see Corollary~\ref{th:passivity}.
	Then, by multiplying each side of~\eqref{th:very_strictly_passivity::cond2} by $\alpha(\omega)$, we obtain the second order inequality
	\begin{equation*}
		\gamma_1 \alpha(\omega)^2
		-
		r(\omega)\cos(\vartheta(\omega)) \alpha(\omega)
		+
		\gamma_2
		\le
		0
	\end{equation*}
	to be satisfied for real $\alpha(\omega) > 0$ with $\cos(\vartheta(\omega)) > 0$.
	Then, the inequality gives real solutions within the set
	\begin{equation*}
		\ps*{
			\frac{r(\omega)\cos(\vartheta(\omega))}{2\gamma_1} - \frac{\sqrt{\delta(\omega)}}{2\gamma_1},
			\frac{r(\omega)\cos(\vartheta(\omega))}{2\gamma_1} + \frac{\sqrt{\delta(\omega)}}{2\gamma_1}
		}
	\end{equation*}
	provided that $\delta(\omega) \ge 0$, hence the claim.
\end{proof}

\subsection{A Weak Form of the Superposition Principle for Nonlinear Frequency Response}

The nonlinear frequency response~\eqref{eq:def_freq_resp} allows the exact characterization of the $\omega$-gain, $\omega$-phase, and $\omega$-radius of a nonlinear system under the input generated by~\eqref{eq:driving_sys} for each $\omega \in \mathcal{N}$.
A defining trait of the frequency response in an LTI system is that its response to a sum of harmonic inputs can be determined from the response to each harmonic, due to the superposition principle.
Consequently, the magnitude and phase response to multiple inputs can be obtained by summing the individual magnitude and phase contributions of each harmonic component.
Yet, for nonlinear systems, the superposition principle breaks down.
As a result, extending the analysis to sums of inputs of the form~\eqref{eq:driving_sys}, each at a distinct $\omega \in \mathcal{N}$, requires a more in-depth analysis.

A formal nonlinear multi-input frequency response analysis necessitates lifting the state of the signal generator to a higher-dimensional space.
To begin with, we rewrite the signal generator~\eqref{eq:driving_sys} introducing a state indexed by $\omega \in \mathcal{N}$, that is
\begin{equation} \label{eq:full_signal_generator}
\begin{aligned}
	\dot{z}_\omega & = s(\omega,z_\omega), &
	 z_\omega(0) & \in \mathcal{Z}_\omega,
\end{aligned}
\end{equation}
where the trajectory $z_\omega(t)$ evolves independently according to the associated $\omega \in \mathcal{N}$ and the mapping $s$ is the same as in~\eqref{eq:driving_sys}.
The resulting state obtained for every $\omega\in\mathcal{N}$ yields a lifted system over the space $\mathcal{Z}_{\mathcal{N}} \coloneq \bigtimes_{\omega\in\mathcal{N}} \mathcal{Z}_\omega$ where the state is the collection of all individual states of~\eqref{eq:full_signal_generator}.
Therefore, the state is defined as $z_\mathcal{N} \coloneq \pr*{ z_\omega }_{\omega\in\mathcal{N}} \in \mathcal{Z}_{\mathcal{N}}$ whereas the lifted system is described by the mapping $\mathcal{S}(z_\mathcal{N}) \coloneq \pr*{s(\omega,z_\omega)}_{\omega\in\mathcal{N}}$.
To restrict the analysis relative to a finite weighted sum of multiple inputs, where each input corresponds to a specific $\omega\in\mathcal{N}$ and weight $M \in \R^{r\times r}$, we introduce the finite subset $\mathcal{A} \subset \mathcal{N} \times \R^{r\times r}$ containing the parameters and the corresponding weights under analysis.
To streamline the notation, we define $\bar{\mathcal{A}} \coloneq \pc*{ \omega \st (\omega, M) \in \mathcal{A} } \subset \mathcal{N}$.
Then, the restriction of the full lifted system over $\mathcal{A}$ gives rise to a \emph{multi-signal generator} on the space $\mathcal{Z}_\mathcal{A}= \bigtimes_{\omega\in\bar{\mathcal{A}}} \mathcal{Z} _\omega$ described by
\begin{equation} \label{eq:multi_driving_sys}
	\begin{aligned}
		\dot{z}_\mathcal{A}
		&
		= \mathcal{S}(z_\mathcal{A}), \quad z_\mathcal{A}(0) \in \mathcal{Z}_{\mathcal{A}}
		\\
		u_{\mathcal{A}}
		&
		= \sum_{(\omega,M)\in\mathcal{A}} M \ell(\omega, z_\omega).
	\end{aligned}
\end{equation}

The multi-signal generator~\eqref{eq:multi_driving_sys} formally characterizes the finite sum of multiple weighted inputs of the family~\eqref{eq:full_signal_generator}, where each input $\ell(\omega, z_\omega)$ of parameter $\omega$ is weighted by $M$, where $(\omega, M) \in \mathcal{A}$.
To proceed with the multi-input analysis, we first need to extend Assumptions~\ref{ass:z_u_periodic} and~\ref{ass:compact_invariant_ultimately_bounded} to establish guarantees for the interconnected system~\eqref{eq:multi_driving_sys}-\eqref{eq:driven_sys}.
\begin{assumption} \label{ass:z_u_A_periodic}
	The set $\mathcal{A} \in \mathcal{N} \times \R$ is such that the state $z_{\mathcal{A}}$ and output $u_{\mathcal{A}}$ of~\eqref{eq:multi_driving_sys} are continuous and periodic, with period $T_{\mathcal{A}} \in \R_+$.
\end{assumption}
\begin{assumption} \label{ass:multi_input_manifold}
	The set $\mathcal{A} \in \mathcal{N} \times \R$ is such that, for every initial condition $(z_\mathcal{A}(0), x(0)) \in \mathcal{Z}_\mathcal{A} \times \mathcal{X}$ the subset $\mathcal{X}$ of the state space of the system~\eqref{eq:driven_sys}, the trajectories of the interconnected system~\eqref{eq:multi_driving_sys}-\eqref{eq:driven_sys} are ultimately bounded, and remain in $\mathcal{Z}_\mathcal{A} \times \mathcal{X}$.
	Moreover, there exists a (sufficiently) smooth mapping $\varphi_\mathcal{A} : \mathcal{Z}_\mathcal{A} \to \R^n$ such that the submanifold
	\begin{equation} \label{eq:inv_manifold_multi}
		\mathcal{M}_\mathcal{A}
		\coloneq
		\set[\big]{
			(z_\mathcal{A}, x) \in \mathcal{Z}_\mathcal{A} \times \mathcal{X} \st x = \varphi_\mathcal{A}(z)
		}
	\end{equation}
	is well-defined and invariant under~\eqref{eq:multi_driving_sys}-\eqref{eq:driven_sys}.
	That is, if $(z_\mathcal{A}(\bar{t}), x(\bar{t})) \in \mathcal{M}_\mathcal{A}$, then $(z_\mathcal{A}(t), x(t)) \in \mathcal{M}_\mathcal{A}$ for all $t \ge \bar{t}$.
\end{assumption}
Building on the preceding rationale and under these assumptions, we now introduce a definition of multi $\omega$-response that explicitly captures the finite sums of inputs in the sense of~\eqref{eq:multi_driving_sys}.
\begin{definition}[\emph{Multi $\omega$-response}] \label{def:multi_w-response}
	Given the nonlinear system~\eqref{eq:driven_sys} the \emph{multi $\omega$-response} to the input~\eqref{eq:multi_driving_sys} is given by
	\begin{equation}\label{eq:multi_w-response}
		Y_\mathcal{A}(z_\mathcal{A})
		\coloneq
		h \pr*{ \varphi_\mathcal{A}(z_\mathcal{A}), u_\mathcal{A} }.
	\end{equation}
\end{definition}
Recall that, by Assumptions~\ref{ass:z_u_A_periodic} and~\ref{ass:multi_input_manifold}, $z_\mathcal{A}$, $u_\mathcal{A}$ and $Y_{\mathcal{A}}(z_{\mathcal{A}})$ are continuous $T_{\mathcal{A}}$-periodic functions.
Hence, $u_\mathcal{A}$ and $Y_{\mathcal{A}}(z_{\mathcal{A}})$ are bounded and square-integrable.
Then Theorem~\ref{th:weak_super:gain} formalizes a relation between the $\mathcal{L}_2$-norm over $\T_\mathcal{A}$ and the $\omega$-gain $\alpha(\omega)$ associated with every $\omega\in\bar{\mathcal{A}}$.

\begin{theorem} \label{th:weak_super:gain}
	Let $\mathcal{A} \subset \mathcal{N} \times \R^{r\times r}$ be such that there exists a $\bar{\omega} \in \bar{\mathcal{A}}$ such that $\alpha(\bar{\omega}) > 0$.
	Then, there exist one $b_{\omega, M} \in \R$ for each $(\omega, M) \in \mathcal{A}$ such that
	\begin{equation} \label{eq:gain_bound}
	\norm[\big]{Y_\mathcal{A}(z_\mathcal{A})}_{\mathcal{L}_2 \pr*{\T_\mathcal{A}}}
	\le
	\sum_{(\omega, M) \in \mathcal{A}}
	b_{\omega, M} \norm*{M} \alpha(\omega).
	\end{equation}
\end{theorem}
\begin{proof}
	By the boundedness of the multi $\omega$-response, there always exists a $K_{\omega, M}$, for every $(\omega, M) \in \mathcal{A}$, such that
	\begin{equation}\label{eq:proof_bounded-Y}
		-\varsigma_\mathcal{A}
		\le
		Y_\mathcal{A}(z_\mathcal{A})
		\le
		\varsigma_\mathcal{A} 
	\end{equation}
	where
	\begin{equation*}
		\varsigma_\mathcal{A}
		\coloneq
		\sum_{(\omega, M) \in \mathcal{A}}
		K_{\omega, M} \norm*{M} Y(\omega,z_\omega)
		\in \mathcal{L}_2 \pr*{\T_\mathcal{A}, \R^r}.
	\end{equation*}
	For all $\omega \in \bar{\mathcal{A}}$, there exist a positive $W_\omega \in \R^+$ such that
	\begin{equation*}
		\norm*{Y(\omega,z_\omega)}_{\mathcal{L}_2 \pr*{\T_\mathcal{A}}}
		=
		W_\omega \norm*{Y(\omega,z_\omega)}_{\mathcal{L}_2 \pr*{\T_\omega}},
	\end{equation*}
	because $Y_\mathcal{A}(z_\mathcal{A})$ and $Y(\omega,z_\omega)$ are bounded by assumptions.
	By using the matrix induced $2$-norm, we define $Q_{\omega, M} \coloneq \norm[\big]{K_{\omega, M}}$.
	By applying the triangle inequality and then invoking the Cauchy-Schwarz inequality, we deduce that the $\mathcal{L}_2$-norm of $Y_\mathcal{A}(z_\mathcal{A})$ over $\T_{\mathcal{A}}$ satisfies
	\begin{align*}
		&
		\norm[\big]{Y_\mathcal{A}(z_\mathcal{A})}_{\mathcal{L}_2 \pr*{\T_\mathcal{A}}}
		\\
		&
		\le
		\Bigg|\sum_{(\omega,M) \in \mathcal{A}}
		K_{\omega, M} \norm*{M} Y(\omega,z_\omega)\Bigg|_{\mathcal{L}_2 \pr*{\T_\mathcal{A}}}
		\hspace{-18ex}
		\\
		&
		\le
		\sum_{(\omega,M) \in \mathcal{A}}
		\norm[\big]{K_{\omega, M}}
		\norm[\big]{M}
		\norm[\big]{ Y(\omega,z_\omega)}_{\mathcal{L}_2 \pr*{\T_\mathcal{A}}}\\
		&
		=
		\sum_{(\omega,M) \in \mathcal{A}}
		Q_{\omega, M} W_\omega
		\norm[\big]{M}
		\norm[\big]{Y(\omega,z_\omega)}_{\mathcal{L}_2 \pr*{\T_\omega}} .
	\end{align*}
	By defining $\sigma_{\mathcal{A}} \coloneq \sup_{\omega \in \bar{\mathcal{A}}} \, \norm[\big]{\ell(\omega,z_\omega)}_{\mathcal{L}_2 \pr*{\T_\omega}}$, we have
	\begin{equation*}
		\norm[\big]{Y(\omega,z_\omega)}_{\mathcal{L}_2 \pr*{\T_\omega}}
		\le
		\sigma_{\mathcal{A}}
		\frac
		{ \norm[\big]{Y(\omega,z_\omega)}_{\mathcal{L}_2 \pr*{\T_\omega}} }
		{ \norm[\big]{\ell(\omega,z_\omega)}_{\mathcal{L}_2 \pr*{\T_\omega}} }
		=
		\sigma_{\mathcal{A}} \alpha(\omega).
	\end{equation*}
	Therefore, we conclude that
	\begin{equation*}
		\norm[\big]{Y_\mathcal{A}(z_\mathcal{A})}_{\mathcal{L}_2 \pr*{\T_\mathcal{A}}}
		\le
		\sum_{(\omega,M) \in \mathcal{A}}
		\sigma_{\mathcal{A}} Q_{\omega, M} W_\omega \norm[\big]{M} \alpha(\omega) 
	\end{equation*}
	which implies~\eqref{eq:gain_bound} for every $b_{\omega,M} \ge \sigma_{\mathcal{A}} Q_{\omega, M} W_\omega$.
\end{proof}

\begin{theorem}
	Let $\mathcal{A} \subset \mathcal{N} \times \R^{r\times r}$ be such that $\alpha(\bar{\omega}) > 0$ for all $\bar{\omega} \in \bar{\mathcal{A}}$.
	Then, there exist one $c_{\omega, M} \in \R$ and one $d_{\omega, M} \in \R$ for each $(\omega, M) \in \mathcal{A}$ such that
	\begin{align}
		\norm*{ \inner[\big]{u_\mathcal{A}}{Y_\mathcal{A}}_{\mathcal{L}_2 \pr*{\T_\mathcal{A}}} }
		&
		\le
		\sum_{(\omega, M) \in \mathcal{A}}
		c_{\omega, M} \norm[\big]{ r(\omega)\cos(\vartheta(\omega)) },
		\label{eq:phase_bound_u}
		\\
		\norm*{ \inner[\big]{\dot{u}_\mathcal{A}}{Y_\mathcal{A}}_{\mathcal{L}_2 \pr*{\T_\mathcal{A}}} }
		&
		\le
		\sum_{(\omega, M) \in \mathcal{A}}
		d_{\omega, M} \norm[\big]{ r(\omega) \sin(\vartheta(\omega)) },
		\label{eq:phase_bound_udot}
	\end{align}
\end{theorem}

\begin{proof}
	Since the proofs of~\eqref{eq:phase_bound_u} and~\eqref{eq:phase_bound_udot} follow the same rationale, we proceed by proving only~\eqref{eq:phase_bound_u}.
	Using the bounds~\eqref{eq:proof_bounded-Y} on the multi $\omega$-response, we have that the Euclidean norm of $\inner[\big]{u_\mathcal{A}}{Y_\mathcal{A}}_{\mathcal{L}_2 \pr*{\T_\mathcal{A}}}$ satisfies
	\begin{align*}
		\norm*{ \inner[\big]{u_\mathcal{A}}{Y_\mathcal{A}}_{\mathcal{L}_2 \pr*{\T_\mathcal{A}}} }
		\le
		\norm*{
			\inner[\big]
			{u_\mathcal{A}}
			{\sum_{(\omega, M) \in \mathcal{A}} K_{\omega, M} \norm*{M} Y(\omega,z_\omega)}
			_{\mathcal{L}_2 \pr*{\T_\mathcal{A}}}.
		}
	\end{align*}
	Then, by substituting $u_\mathcal{A}$ in the series form in~\eqref{eq:multi_driving_sys} and invoking the Cauchy-Schwarz inequality, the right-hand side of the above inequality can be further bounded as follows
	\begin{align*}
		&
		\norm*{
			\inner[\big]
			{u_\mathcal{A}}
			{\sum_{(\omega, M) \in \mathcal{A}} K_{\omega, M} \norm*{M} Y(\omega,z_\omega)}
			_{\mathcal{L}_2 \pr*{\T_\mathcal{A}}}
		}
		\\
		&
		\hspace{3ex}
		\le
		\sum_{(\omega, M) \in \mathcal{A}}
		\tilde{K}_{\omega, M}
		\sum_{(\tilde{\omega},\tilde{M})\in\mathcal{A}}
		\norm*{\tilde{M}}
		\norm*{ \inner[\big]{ \ell(\tilde{\omega}, z_{\tilde{\omega}})}{Y(\omega,z_\omega)}_{\mathcal{L}_2 \pr*{\T_\mathcal{A}}} },
	\end{align*}
	where $\tilde{K}_{\omega, M} \coloneq \norm*{K_{\omega, M}} \norm*{M}$.
	Since the inner product above is always bounded, then for all $(\tilde{\omega},\tilde{M}) \in \mathcal{A}$ there always exists a $P_{\omega,M}\in\R$ such that
	\begin{align*}
		&
		\norm*{\tilde{M}}
		\norm*{ \inner[\big]{ \ell(\tilde{\omega}, z_{\tilde{\omega}})}{Y(\omega,z_\omega)}_{\mathcal{L}_2 \pr*{\T_\mathcal{A}}} }
		\\
		&
		\hspace{20ex}
		\le
		\frac{P_{\omega,M}}{\tilde{K}_{\omega, M}}
		\norm*{ \inner[\big]{ \ell(\omega, z_{\omega})}{Y(\omega,z_\omega)}_{\mathcal{L}_2 \pr*{\T_\mathcal{A}}} }.
	\end{align*}
	For all $\omega \in \bar{\mathcal{A}}$, there exist a positive $R_\omega \in \R^+$ such that
	\begin{equation*}
		\inner[\big]{ \ell(\tilde{\omega}, z_{\tilde{\omega}})}{Y(\omega,z_\omega)}_{\mathcal{L}_2 \pr*{\T_\mathcal{A}}}
		=
		R_\omega \inner[\big]{ \ell(\tilde{\omega}, z_{\tilde{\omega}})}{Y(\omega,z_\omega)}_{\mathcal{L}_2 \pr*{\T_\omega}},
	\end{equation*}
	because $Y_\mathcal{A}(z_\mathcal{A})$ and $Y(\omega,z_\omega)$ are bounded by assumptions.
	Therefore, using the properties above, the cardinality $\norm*{\mathcal{A}}$, and defining the constant
	\begin{equation*}
		\sigma_{\mathcal{A}}
		\coloneq
		\sup_{\omega \in \bar{\mathcal{A}}}
		\pr*{
			\norm[\big]{\ell(\omega,z_\omega)}_{\mathcal{L}_2 \pr*{\T_\omega}} \norm[\big]{Y(\omega,z)}_{\mathcal{L}_2 \pr*{\T_\omega}}
		},
	\end{equation*}
	we have that
	\begin{align*}
		\norm*{ \inner[\big]{u_\mathcal{A}}{Y_\mathcal{A}}_{\mathcal{L}_2 \pr*{\T_\mathcal{A}}} }
		\hspace{-5ex}
		&
		\\
		&
		\le
		\norm*{\mathcal{A}} R_\omega P_{\omega,M}
		\norm*{ \inner[\big]{ \ell(\omega, z_{\omega})}{Y(\omega,z_\omega)}_{\mathcal{L}_2 \pr*{\T_\omega}} }
		\\
		&
		\le
		\norm*{\mathcal{A}} R_\omega P_{\omega,M} \sigma_{\mathcal{A}} \norm*{ r(\omega)\cos(\vartheta(\omega)) },
	\end{align*}
	which implies~\eqref{eq:phase_bound_u} for every $c_{\omega, M} \ge \norm*{\mathcal{A}} R_\omega P_{\omega,M} \sigma_{\mathcal{A}}$.
\end{proof}

\section{Loop Shaping Design for Nonlinear Systems}
\condVSpace{-0.25em}
\label{sec: loop-shaping}

The essence of the loop-shaping synthesis in the frequency domain is that closed-loop performances can be attained by designing a feedback law based upon prescribed steady-state characteristics of the open-loop frequency response, see, \emph{e.g.},~\cite[Ch.~2]{skogestad2005multivariable} and~\cite{mcFarlane1992loop}.
Following a similar idea, loop-shaping in the nonlinear setting can be achieved by constructing a parameterized feedback law yielding a parameterized $\omega$-response.
To this end, consider nonlinear systems modeled by equations of the form
\begin{equation}\label{eq:closed_loop}
\begin{aligned}
	\dot{x}_c & = f_c(x_c,u,v), &
	y_c & = h(x_c,u),
\end{aligned}
\end{equation}
where $f_c$ is such that $f_c(x_c,u,0) = f(x,u)$, and $v\in\R^{\bar{m}}$ ($\bar{m} \in \N$) is a control input that is used to model a state feedback controller. 
The control input can be cast as a dynamic state feedback of the form
\begin{equation} \label{eq:feedback}
\begin{aligned}
	\dot{x}_p &= f_p(x_p,x_c,u), &
	v &= \kappa(x_p,x_c,u),
\end{aligned}
\end{equation}
where $f_p : \R^{\bar{n}} \times \R^n \times \R^{m} \to \R^{\bar{n}} $ and $\kappa : \R^{\bar{n}} \times \R^n \times \R^{m} \to \R^{\bar{m}} $ is a nonlinear function used to achieve stabilization with prescribed closed-loop characteristics.
The forced response of the feedback system~\eqref{eq:driven_sys} to the input~\eqref{eq:driving_sys} is thus given by the cascade feedback system
\begin{equation}\label{eq:cascade_sys_loop}
\begin{aligned}
	\dot{z} & = s(\omega, z), \\
	\dot{x}_c & = f_c(x_c,\ell(\omega, z),v), \\
	\dot{x}_p &= f_p(x_p,x_c,\ell(\omega, z)), \\
	v &= \kappa(x_p,x_c,\ell(\omega, z)),\\
	y_c & = h(x_c,\ell(\omega, z)),
\end{aligned}
\end{equation}
associated with the invariant manifold given by the graph $$\mathcal{M}_c \coloneq \set*{(\omega, z, x_c,x_p) \st x_c = \varphi_c(\omega, z), \ x_p = \varphi_p(\omega, z) }$$ of (sufficiently) smooth mappings $x_c = \varphi_c(\omega, z)$ and $x_p = \varphi_p(\omega, z)$.
Thereby, the $\omega$-response to the input~\eqref{eq:driving_sys} for all $\omega\in\diff{\mathcal{N}}$ is given by
\begin{equation} \label{eq:w-response_cl}
	Y_c(\omega,z) \coloneq h(\varphi_c(\omega,z),\ell(\omega,z)).
\end{equation}
The resulting $\omega$-response function~\eqref{eq:w-response_cl} associated with the feedback system is then parameterized by the choice of the functions $f_p$ and $\kappa$.
This parameterization gives degrees of freedom in the selection of a desired shape of the closed loop frequency response.
The frequency response function of the system~\eqref{eq:closed_loop}-\eqref{eq:feedback} to the input~\eqref{eq:driving_sys} is referred to as $\Gamma_c(\omega)$ with associated $\omega$-gain, $\omega$-phase, and $\omega$-radius functions denoted for convenience as $(\alpha_c(\omega)$, $\vartheta_c(\omega)$, and $r_c(\omega))$, respectively.
To define stabilization with prescribed closed-loop characteristics, we formulate a loop-shaping problem as a constrained design problem on the desired frequency response function.

\begin{figure*}[t]
	\centering
	\condVSpace{-1.5em}
	\begin{subfigure}{0.325\linewidth}
	\centering
	\includegraphics[width=\linewidth]{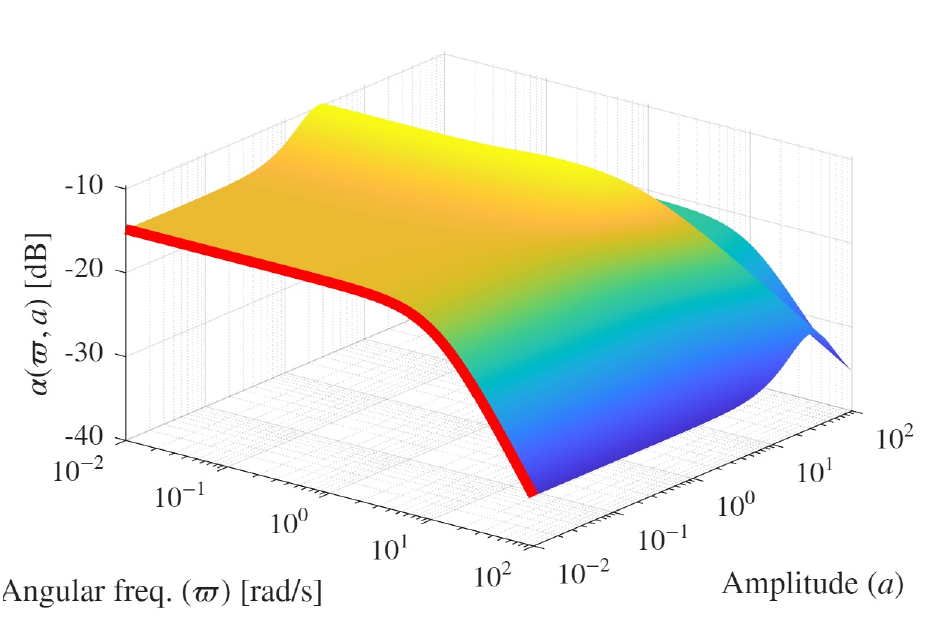}
	\condVSpace{-1.5em}
	\caption{$\omega$-gain}
	\label{fig:bode:ex_IFAC_state::gain}
	\end{subfigure}
	\hfill
	\begin{subfigure}{0.325\linewidth}
	\centering
	\includegraphics[width=\linewidth]{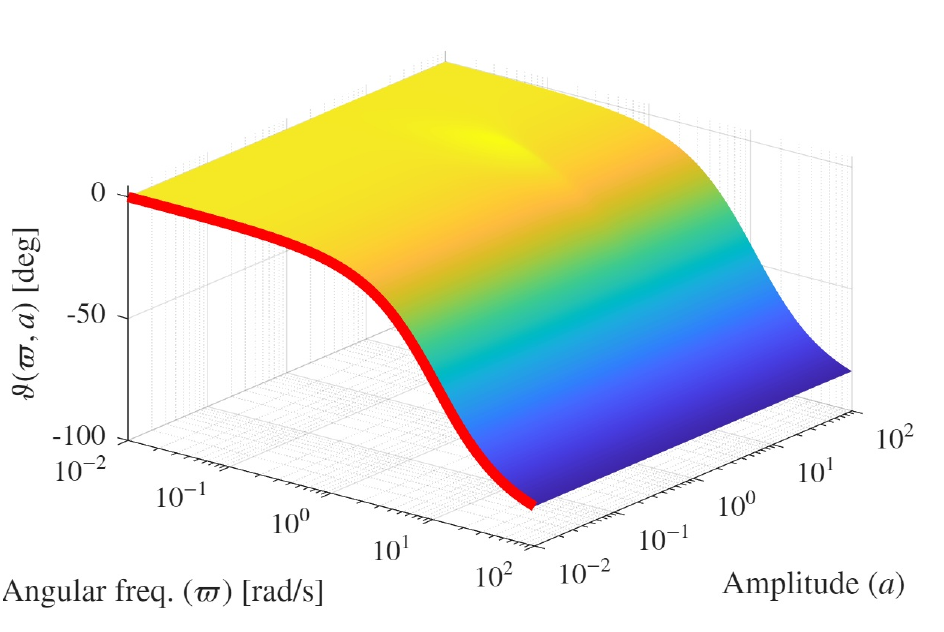}
	\condVSpace{-1.5em}
	\caption{$\omega$-phase}
	\label{fig:bode:ex_IFAC_state::phase}
	\end{subfigure}
	\hfill
	\begin{subfigure}{0.325\linewidth}
	\centering
	\includegraphics[width=\linewidth]{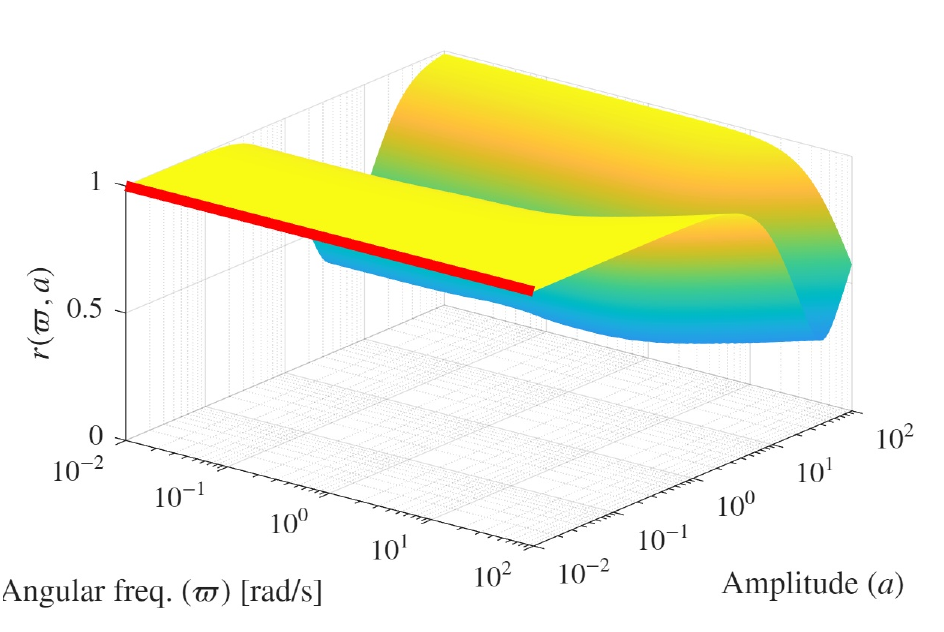}
	\condVSpace{-1.5em}
	\caption{$\omega$-radius}
	\label{fig:bode:ex_IFAC_state::radius}
	\end{subfigure}

	\caption{
		Bode-like diagram of the nonlinear system~\eqref{ex:feedback} under state-feedback~\eqref{eq:state_feedback} with $K=10$, angular frequency $\varpi \in [10^{-2}, 10^2]$, and amplitude $a_u \in [10^{-2}, 10^2]$.
		The red curve corresponds to the $\omega$-gain, $\omega$-phase, and $\omega$-radius associated with the linearized system.
	}
	\label{fig:bode:ex_IFAC_state}
	\condVSpace{-1em}
\end{figure*}

\begin{figure*}[t]
	\centering
	\begin{subfigure}{0.325\linewidth}
	\centering
	\includegraphics[width=\linewidth]{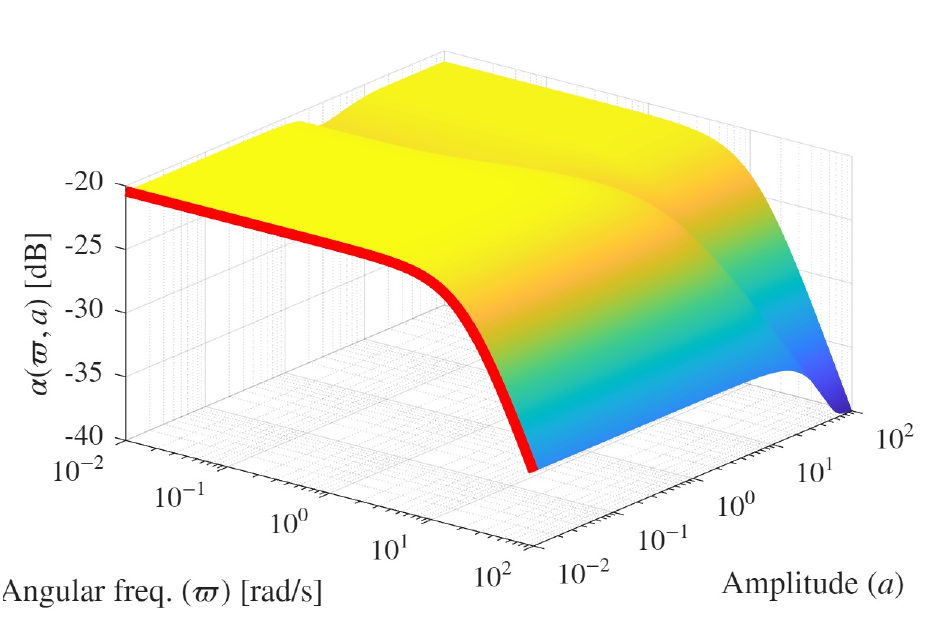}
	\condVSpace{-1.5em}
	\caption{$\omega$-gain}
	\label{fig:bode:ex_IFAC_output::gain}
	\end{subfigure}
	\hfill
	\begin{subfigure}{0.325\linewidth}
	\centering
	\includegraphics[width=\linewidth]{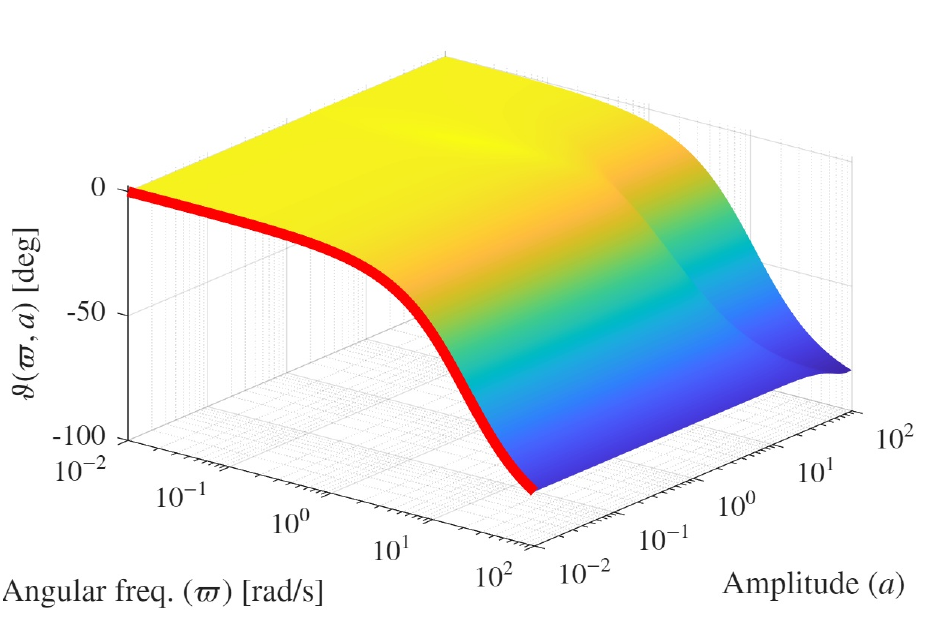}
	\condVSpace{-1.5em}
	\caption{$\omega$-phase}
	\label{fig:bode:ex_IFAC_output::phase}
	\end{subfigure}
	\hfill
	\begin{subfigure}{0.325\linewidth}
	\centering
	\includegraphics[width=\linewidth]{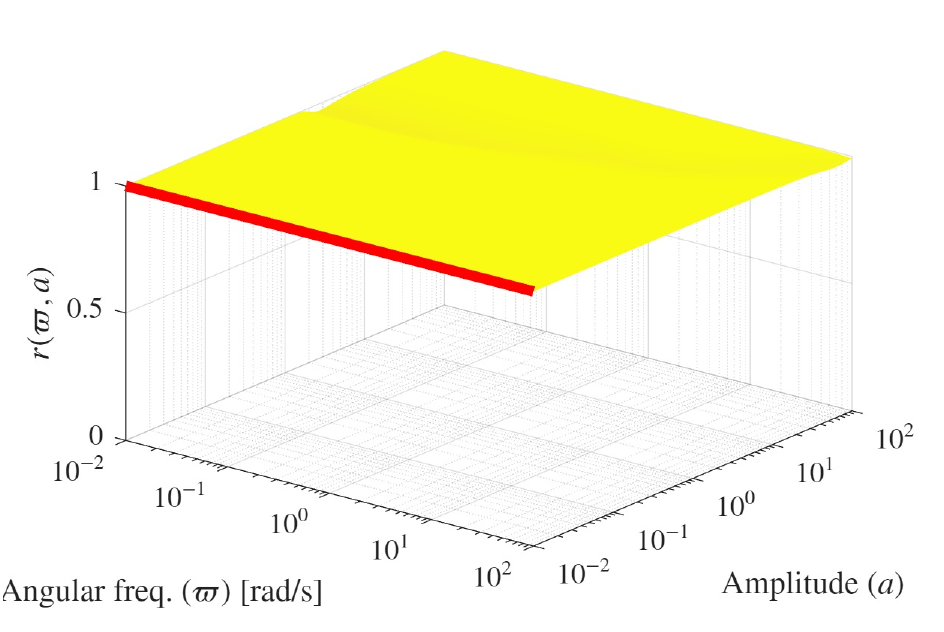}
	\condVSpace{-1.5em}
	\caption{$\omega$-radius}
	\label{fig:bode:ex_IFAC_output::radius}
	\end{subfigure}
	\caption{
	Bode-like diagram of the nonlinear system~\eqref{ex:feedback} under output-feedback~\eqref{eq:output_feedback} with $K=10$, angular frequency $\varpi \in [10^{-2}, 10^2]$, and amplitude $a_u \in [10^{-2}, 10^2]$.
	The red curve corresponds to the $\omega$-gain, $\omega$-phase, and $\omega$-radius associated with the linearized system about the equilibrium $x = 0$.
	}
	\label{fig:bode:ex_IFAC_output}
	\condVSpace{-1.5em}
\end{figure*}

\begin{description}

	\item[The Nonlinear Loop-Shaping Problem.]
	Given system~\eqref{eq:closed_loop} driven by the input~\eqref{eq:driving_sys}, and a closed subset of desired specifications $\mathcal{SP}$, design the feedback~\eqref{eq:feedback} such that

	\begin{enumerate}[label={(R\arabic*)}, ref={(R\arabic*)}]
		
		\item \label{en:ls:stable}
		the origin, $(x_c,x_p) = (0,0)$, of the dynamics
		\begin{align*}
			\dot{x}_c & = f_c(x_c,0,\kappa(x_p,x_c,0)) \\
			\dot{x}_p & = f_p(x_p,x_c,0)
		\end{align*}
		is an asymptotically stable equilibrium,
		
		\item \label{en:ls:gain}
		the frequency response function $\Gamma_c : \mathcal{N} \to \C$ of the system~\eqref{eq:closed_loop}-\eqref{eq:feedback} to the input~\eqref{eq:driving_sys} is such that
		\begin{align*}
			(\alpha_c(\omega), \vartheta_c(\omega), r_c(\omega)) 
			\in
			\mathcal{SP},
			\hspace{3ex}
			\forall \omega\in  \diff{\mathcal{N}},
			\hspace{-3ex}
		\end{align*}
		where $\mathcal{SP} \subset \R^3$ is the \emph{specification subset} defined as
		$\mathcal{SP}
		\coloneq
		\ps*{ \alpha_{\min}, \alpha_{\max} } \times
		\ps*{ \vartheta_{\min}, \vartheta_{\max} } \times
		\ps*{ r_{\min}, r_{\max} }$.
	\end{enumerate}
\end{description}

In addition to~\ref{en:ls:stable}, which is the standard asymptotic stabilization requirements~\cite[Sec. 8.2]{isidori1995nonlinear}, the nonlinear loop-shaping problem introduced in~\ref{en:ls:gain} requires a polyhedral set of specifications $\mathcal{SP}$ which constrain the shape of the frequency response function through the appropriate selection of the feedback~\eqref{eq:feedback}.
These specifications can be set and validated through the associated nonlinear Bode gain diagram.
This formulation of the loop-shaping problem yields a complementary role when combined with other nonlinear control laws, see~\cite{khalil2015nonlinear}.
For instance, the $\mathcal{H}_\infty$ control of~\cite{isidori1992disturbance} can be recast in a static loop-shaping synthesis and thus enriched by imposing specific requirements on the shape of the $\omega$-gain function.
This complementary role can be achieved by taking advantage of the degrees of freedom in the controller, whether static or dynamic, to provide qualitative performance without compromising stability.

\begin{example}
\label{ex:IFAC}

Let~\eqref{eq:closed_loop} be a state-feedback version of the nonlinear system described by~\eqref{eq:nonlinear-system_example_f}, that is
\begin{equation} \label{ex:feedback}
	f_c(x,u,v)
	=
	\begin{pmatrix}
		x_1 + u \\
		-x_2 + x_1u
	\end{pmatrix}
	+ \kappa(x,u) ,
\end{equation}
with output function $h(x,u) = x_1 +\tanh(x_1 + x_2)$, and static feedback $\kappa(x,u)$ of the form
\begin{equation} \label{eq:state_feedback}
	\kappa(x,u)
	=
	\begin{pmatrix}
		-Kx_1 \\
		Kx_1^2
	\end{pmatrix}.
\end{equation}
Note that the feedback system~\eqref{ex:feedback} satisfies the stability requirement~\ref{en:ls:stable} for every $K>1$.
Let the signal generator~\eqref{eq:driving_sys} be of the form~\eqref{eq:LTI-exosystem} with parameters $\omega = (\varpi, a_u) \in \R_+^2$, state $z = (z_1, z_2) \in \R^2$, and $z_0(\omega) = (0, a_u) \in \R^2$.

The design specification requires input attenuation.
Therefore, we can choose the parameter $K$ appropriately to guarantee that the $\omega$-gain is such that the output attenuates the effect of the input, that is $\alpha_{\max} < 1$.
Figure~\ref{fig:bode:ex_IFAC_state} illustrates the Bode-like diagram of the closed-loop system~\eqref{ex:feedback} with the feedback~\eqref{eq:state_feedback} relative to the choice of $K=10$.
Although this state-feedback modifies the shape of the surface on the amplitude axis, we note that, in both cases, attenuation of the input is achieved.
At the same time, the shape of the $\omega$-phase is largely consistent with the linearized model within the bounds of the selected amplitudes, while the $\omega$-radius exhibits a hollow shape for certain amplitudes, indicating high distortion of the output signal compared to the input.

However, suppose that achieving stabilization requires only a small distortion in the output, meaning that the $\omega$-radius under study remains close to $1$. Hence, we consider analyzing the effect of a negative output feedback of the form
\begin{equation}\label{eq:output_feedback}
	\kappa(x,u)
	=
	-K
	\begin{pmatrix}
		x_1 +\tanh(x_1 + x_2)\\
		0
	\end{pmatrix}.
\end{equation}

Figure~\ref{fig:bode:ex_IFAC_output} illustrates the Bode-like diagram of the closed loop system~\eqref{ex:feedback} with the feedback~\eqref{eq:output_feedback} relative to the choice of $K=10$.
As we can see, the output feedback still guarantees attenuation of the effect of the input, with a similar shape of the $\omega$-phase function of the state-feedback but with an $\omega$-radius function which almost flattens out over all frequencies and amplitudes, providing uniform consistency in the shape of the output with respect to the input $u$.
\end{example}

\section{Conclusion and Outlooks} 
\label{sec:end}

Motivated by the question of whether the time-domain interpretation of the frequency response used in linear theory could be enhanced analogously for nonlinear systems, this article introduced a framework for nonlinear frequency response.
We showed that, to bridge the gap between linear and nonlinear analysis, a quantitative measure of the distortion of the shape of the input through the system is necessary, along with suitable notions of gain and phase defined to deal with nonlinear periodic excitations.
Together, these functions provide a complete characterization of the frequency response, which can be illustrated graphically by nonlinear enhancement of the Bode diagrams.
Additionally, an analysis of the frequency responses of various dissipative systems with different supply rates has been conducted.
Then, a nonlinear multi-input frequency response analysis was performed, which was found to be a weak form of the superposition principle for nonlinear systems.
Finally, graphical illustrations of the frequency response allow the nonlinear loop-shaping problem to be formulated in terms of a specification set that accounts for desired quantitative requirements on the feedback system.

Bode-like diagrams introduced here provide an exact quantitative measure of the frequency-amplitude relation between input and output, extending beyond linear analysis.
The Hartman-Grobman theorem implies the existence of a neighborhood of the hyperbolic equilibrium and an input bound such that all trajectories starting in the neighborhood remain in it.
Although input bounds are not easily quantifiable, the proposed graphical representation offers a quantitative illustration of the point at which such dynamics transition to nonlinear regimes.
The nonlinear frequency domain framework can be seamlessly employed in a data-driven fashion when only input and output measurements are available~\cite{moreschini2025moment}.

This nonlinear frequency response framework comes with several open questions worth addressing in future work.
First, we aim to clarify the interrelationship between magnitude and phase of the frequency response, in the spirit of the Kramers-Kronig relations~\cite{bechhoefer2011kramers}, which connect the real and imaginary parts of the frequency response function~\eqref{eq:tiangular_frequency}.
This would allow giving a one-to-one correspondence between small gain and small phase~\cite{chen2021small}.
Following the line of~\cite{moreschini2024generalized}, this framework can be extended to include abstract time domains, allowing for the consideration of other systems, such as discrete time and hybrid settings.
Moreover, we aim to relate the $\omega$-response and its Bode-like diagrams with the stability of the system under feedback.
In particular, for small amplitudes, stability can be reconnected with the Bode stability criterion.
Nevertheless, the theoretical relationship between stability conditions and the properties of the $\omega$-response function for amplitudes beyond the local regime is yet to be understood.
In doing so, we should look more closely at how the poles of a nonlinear system~\cite{padoan2017eigenvalues} are reflected in its frequency response.
Finally, the proposed framework enables graphical validation of nonlinear interpolation methods, \emph{e.g.}, moment matching~\cite{astolfi2010model,moreschini2025closed,scarciotti2024survey}, and the design of surrogate models with provable performance at different frequencies.
Specifically, exploiting the $\omega$-radius function, it is possible to devise approximate realizations that interpolate the system at targeted frequencies, ensuring low complexity in all other operational regimes.

\condVSpace{-0.5em}

\section*{Acknowledgment}
We would like to thank A. Astolfi for his insightful discussions during the initial stages of this research, as well as for his valuable suggestions on the preliminary drafts.

\condVSpace{-0.5em}

\iftoggle{paperSC}
{}
{\section*{References}}
\condVSpace{-1.5em}
\bibliographystyle{IEEEtran}
\bibliography{ref}

\begin{thebibliography}{10}
\providecommand{\url}[1]{#1}
\csname url@samestyle\endcsname
\providecommand{\newblock}{\relax}
\providecommand{\bibinfo}[2]{#2}
\providecommand{\BIBentrySTDinterwordspacing}{\spaceskip=0pt\relax}
\providecommand{\BIBentryALTinterwordstretchfactor}{4}
\providecommand{\BIBentryALTinterwordspacing}{\spaceskip=\fontdimen2\font plus
\BIBentryALTinterwordstretchfactor\fontdimen3\font minus
  \fontdimen4\font\relax}
\providecommand{\BIBforeignlanguage}[2]{{%
\expandafter\ifx\csname l@#1\endcsname\relax
\typeout{** WARNING: IEEEtran.bst: No hyphenation pattern has been}%
\typeout{** loaded for the language `#1'. Using the pattern for}%
\typeout{** the default language instead.}%
\else
\language=\csname l@#1\endcsname
\fi
#2}}
\providecommand{\BIBdecl}{\relax}
\BIBdecl

\bibitem{MacFarlane1979perspectives}
A.~MacFarlane, ``The development of frequency-response methods in automatic
  control [perspectives],'' \emph{IEEE Trans. Autom. Control}, vol.~24, no.~2,
  pp. 250--265, 1979.

\bibitem{baillieul2021encyclopedia}
J.~Baillieul and T.~Samad, Eds., \emph{Encyclopedia of Systems and
  Control}.\hskip 1em plus 0.5em minus 0.4em\relax Springer, 2021.

\bibitem{huang2024gain}
L.~Huang, D.~Wang, X.~Wang, H.~Xin, P.~Ju, K.~H. Johansson, and F.~D{\"o}rfler,
  ``Gain and phase: Decentralized stability conditions for power
  electronics-dominated power systems,'' \emph{IEEE Trans. Power Syst.},
  vol.~39, no.~6, pp. 7240--7256, 2024.

\bibitem{xu2025interharmonic}
W.~Xu, J.~Yong, H.~J. Marquez, and C.~Li, ``Interharmonic power--a new concept
  for power system oscillation source location,'' \emph{IEEE Trans. Power
  Syst.}, 2025.

\bibitem{gibbard2015small}
M.~J. Gibbard, P.~Pourbeik, and D.~J. Vowles, \emph{Small-Signal Stability,
  Control and Dynamic Performance of Power Systems}.\hskip 1em plus 0.5em minus
  0.4em\relax Univ. Adelaide Press, 2015.

\bibitem{gu2022power}
Y.~Gu and T.~C. Green, ``Power system stability with a high penetration of
  inverter-based resources,'' \emph{Proc. IEEE}, vol. 111, no.~7, pp. 832--853,
  2022.

\bibitem{olsman2019hard}
N.~Olsman, A.-A. Baetica, F.~Xiao, Y.~P. Leong, R.~M. Murray, and J.~C. Doyle,
  ``Hard limits and performance tradeoffs in a class of antithetic integral
  feedback networks,'' \emph{Cell Systems}, vol.~9, no.~1, pp. 49--63, 2019.

\bibitem{khammash2022cybergenetics}
M.~H. Khammash, ``Cybergenetics: Theory and applications of genetic control
  systems,'' \emph{Proc. IEEE}, vol. 110, no.~5, pp. 631--658, 2022.

\bibitem{del2015biomolecular}
D.~Del~Vecchio and R.~M. Murray, \emph{Biomolecular Feedback Systems}.\hskip
  1em plus 0.5em minus 0.4em\relax Princeton Univ. Press, 2015.

\bibitem{keener2025mathematical}
J.~Keener and J.~Sneyd, \emph{Mathematical Physiology}.\hskip 1em plus 0.5em
  minus 0.4em\relax Springer, 2025.

\bibitem{paul2007analysis}
C.~R. Paul, \emph{Analysis of Multiconductor Transmission Lines}.\hskip 1em
  plus 0.5em minus 0.4em\relax John Wiley \& Sons, 2007.

\bibitem{bullmore2009complex}
E.~Bullmore and O.~Sporns, ``Complex brain networks: graph theoretical analysis
  of structural and functional systems,'' \emph{Nat. Rev. Neurosci.}, vol.~10,
  no.~3, pp. 186--198, 2009.

\bibitem{selvaratnam2025frequency}
D.~Selvaratnam, A.~Moreschini, A.~Das, T.~Parisini, and H.~Sandberg,
  ``Frequency-domain bounds for the multiconductor telegrapher's equation,''
  \emph{arXiv preprint arXiv:2504.01599}, 2025.

\bibitem{chua2024memristors}
L.~O. Chua, ``Memristors on ‘edge of chaos’,'' \emph{Nat. Rev. Electr.
  Eng.}, vol.~1, no.~9, pp. 614--627, 2024.

\bibitem{ortega2018graph}
A.~Ortega, P.~Frossard, J.~Kova{\v{c}}evi{\'c}, J.~M.~F. Moura, and
  P.~Vandergheynst, ``Graph signal processing: Overview, challenges, and
  applications,'' \emph{Proc. IEEE}, vol. 106, no.~5, pp. 808--828, 2018.

\bibitem{boashash2015time}
B.~Boashash, \emph{Time-Frequency Signal Analysis and Processing: A
  Comprehensive Reference}.\hskip 1em plus 0.5em minus 0.4em\relax Acad. Press,
  2015.

\bibitem{dong2019learning}
X.~Dong, D.~Thanou, M.~Rabbat, and P.~Frossard, ``Learning graphs from data: A
  signal representation perspective,'' \emph{IEEE Signal Process. Mag.},
  vol.~36, no.~3, pp. 44--63, 2019.

\bibitem{skogestad2005multivariable}
S.~Skogestad and I.~Postlethwaite, \emph{Multivariable Feedback Control:
  Analysis and Design}.\hskip 1em plus 0.5em minus 0.4em\relax John Wiley \&
  Sons, 2005.

\bibitem{franklin2025feedback}
G.~F. Franklin, J.~D. Powell, and A.~Emami-Naeini, \emph{Feedback Control of
  Dynamic Systems}.\hskip 1em plus 0.5em minus 0.4em\relax Pearson, 2025.

\bibitem{doyle2013feedback}
J.~C. Doyle, B.~A. Francis, and A.~R. Tannenbaum, \emph{Feedback Control
  Theory}.\hskip 1em plus 0.5em minus 0.4em\relax Courier Corp., 2013.

\bibitem{bode1945network}
H.~W. Bode, \emph{Network Analysis and Feedback Amplifier Design}.\hskip 1em
  plus 0.5em minus 0.4em\relax D. Van Nostrand Co., Inc., 1945.

\bibitem{Nyquist1932regeneration}
H.~Nyquist, ``Regeneration theory,'' \emph{Bell Syst. Tech. J.}, vol.~11,
  no.~1, pp. 126--147, 1932.

\bibitem{zames2003feedback}
G.~Zames, ``Feedback and optimal sensitivity: Model reference transformations,
  multiplicative seminorms, and approximate inverses,'' \emph{IEEE Trans.
  Autom. Control}, vol.~26, no.~2, pp. 301--320, 2003.

\bibitem{tannenbaum1980feedback}
A.~Tannenbaum, ``Feedback stabilization of linear dynamical plants with
  uncertainty in the gain factor,'' \emph{Int. J. Control}, vol.~32, no.~1, pp.
  1--16, 1980.

\bibitem{mcFarlane1992loop}
D.~McFarlane and K.~Glover, ``A loop-shaping design procedure using
  {$H_\infty$} synthesis,'' \emph{IEEE Trans. Autom. Control}, vol.~37, no.~6,
  pp. 759--769, 1992.

\bibitem{bechhoefer2011kramers}
J.~Bechhoefer, ``{Kramers}--{Kronig}, {Bode}, and the meaning of zero,''
  \emph{Am. J. Phys.}, vol.~79, no.~10, pp. 1053--1059, 2011.

\bibitem{isidori1990output}
A.~Isidori and C.~I. Byrnes, ``Output regulation of nonlinear systems,''
  \emph{IEEE Trans. Autom. Control}, vol.~35, no.~2, pp. 131--140, 1990.

\bibitem{isidori1995nonlinear}
A.~Isidori, \emph{Nonlinear Control Systems}.\hskip 1em plus 0.5em minus
  0.4em\relax Springer, 1995.

\bibitem{isidori2008steady}
A.~Isidori and C.~I. Byrnes, ``Steady-state behaviors in nonlinear systems with
  an application to robust disturbance rejection,'' \emph{Annu. Rev. Control},
  vol.~32, no.~1, pp. 1--16, 2008.

\bibitem{isidori2017lectures}
A.~Isidori, \emph{Lectures in Feedback Design for Multivariable Systems}.\hskip
  1em plus 0.5em minus 0.4em\relax Springer, 2017.

\bibitem{george1959continuous}
D.~A. George, ``Continuous nonlinear systems,'' Massachusetts Institute of
  Technology. Research Laboratory of Electronics, Tech. Rep., 1959.

\bibitem{bayma2018analysis}
R.~S. Bayma, Y.~Zhu, and Z.-Q. Lang, ``The analysis of nonlinear systems in the
  frequency domain using nonlinear output frequency response functions,''
  \emph{Automatica}, vol.~94, pp. 452--457, 2018.

\bibitem{boyd2003fading}
S.~Boyd and L.~Chua, ``Fading memory and the problem of approximating nonlinear
  operators with {Volterra} series,'' \emph{IEEE Trans. Circuits Syst.},
  vol.~32, no.~11, pp. 1150--1161, 2003.

\bibitem{isidori1992disturbance}
A.~Isidori and A.~Astolfi, ``Disturbance attenuation and ${H_\infty}$-control
  via measurement feedback in nonlinear systems,'' \emph{IEEE Trans. Autom.
  Control}, vol.~37, no.~9, pp. 1283--1293, 1992.

\bibitem{astolfi2010model}
A.~Astolfi, ``Model reduction by moment matching for linear and nonlinear
  systems,'' \emph{IEEE Trans. Autom. Control}, vol.~55, no.~10, pp.
  2321--2336, 2010.

\bibitem{pavlov2007frequency}
A.~Pavlov, N.~van~de Wouw, and H.~Nijmeijer, ``Frequency response functions for
  nonlinear convergent systems,'' \emph{IEEE Trans. Autom. Control}, vol.~52,
  no.~6, pp. 1159--1165, 2007.

\bibitem{chua1979frequency}
L.~O. Chua and C.-Y. Ng, ``Frequency domain analysis of nonlinear systems:
  general theory,'' \emph{IEE J. Electron. Circuits Syst.}, vol.~3, pp.
  165--185, 1979.

\bibitem{billings1994analysing}
S.~Billings and H.~Zhang, ``Analysing non-linear systems in the frequency
  domain--{II}. the phase response,'' \emph{Mech. Syst. Signal Process.},
  vol.~8, no.~1, pp. 45--62, 1994.

\bibitem{rugh1981nonlinear}
W.~J. Rugh, \emph{Nonlinear System Theory}.\hskip 1em plus 0.5em minus
  0.4em\relax Johns Hopkins Univ. Press, 1981.

\bibitem{zames1966inputII}
G.~Zames, ``On the input-output stability of time-varying nonlinear feedback
  systems--{Part} {II}: Conditions involving circles in the frequency plane and
  sector nonlinearities,'' \emph{IEEE Trans. Autom. Control}, vol.~11, no.~3,
  pp. 465--476, 1966.

\bibitem{chen2021phasenonlinearsystems}
C.~Chen, D.~Zhao, W.~Chen, S.~Z. Khong, and L.~Qiu, ``Phase of nonlinear
  systems,'' \emph{arXiv preprint arXiv:2012.00692}, 2021.

\bibitem{chen2021singularanglenonlinearsystems}
C.~Chen, D.~Zhao, and S.~Z. Khong, ``The singular angle of nonlinear systems,''
  \emph{Automatica}, vol. 181, p. 112515, 2025.

\bibitem{king2009hilbert}
F.~W. King, \emph{{Hilbert} Transforms}.\hskip 1em plus 0.5em minus 0.4em\relax
  Cambridge Univ. Press,, 2009.

\bibitem{chaffey2023graphical}
T.~Chaffey, F.~Forni, and R.~Sepulchre, ``Graphical nonlinear system
  analysis,'' \emph{IEEE Trans. Autom. Control}, vol.~68, no.~10, pp.
  6067--6081, 2023.

\bibitem{ryu2022scaled}
E.~K. Ryu, R.~Hannah, and W.~Yin, ``Scaled relative graphs: Nonexpansive
  operators via {2D} {Euclidean} geometry,'' \emph{Math. Program.}, vol. 194,
  no.~1, pp. 569--619, 2022.

\bibitem{moreschini2026onfrequency}
A.~Moreschini, M.~Scandella, and A.~Astolfi, ``On the frequency response and
  loop shaping for nonlinear systems,'' in \emph{IFAC WC 2026}, (Accepted).

\bibitem{he2006some}
J.-H. He, ``Some asymptotic methods for strongly nonlinear equations,''
  \emph{Int. J. Mod. Phys. B}, vol.~20, no.~10, pp. 1141--1199, 2006.

\bibitem{cveticanin2018strong}
L.~Cveticanin, \emph{Strong Nonlinear Oscillators: Analytical Solutions}.\hskip
  1em plus 0.5em minus 0.4em\relax Springer, 2018.

\bibitem{carr2012a}
J.~Carr, \emph{Applications of Centre Manifold Theory}.\hskip 1em plus 0.5em
  minus 0.4em\relax Springer, 2012.

\bibitem{hartman1982ordinary}
P.~Hartman, \emph{Ordinary Differential Equations}, 2nd~ed.\hskip 1em plus
  0.5em minus 0.4em\relax Birkhäuser, 1982.

\bibitem{chaffey2025amplitude}
T.~Chaffey and F.~Forni, ``Amplitude response and square wave describing
  functions,'' \emph{Eur. J. Control.}, p. 101310, 2025.

\bibitem{Cveticanin2009a}
L.~Cveticanin, ``Oscillator with fraction order restoring force,'' \emph{J.
  Sound Vib.}, vol. 320, no.~4, pp. 1064--1077, 2009.

\bibitem{Olver2010a}
F.~W.~J. Olver, D.~W. Lozier, R.~F. Boisvert, and C.~W. Clark, \emph{NIST
  Handbook of Mathematical Functions}.\hskip 1em plus 0.5em minus 0.4em\relax
  Cambridge Univ. Press, 2010.

\bibitem{willems1972dissipative}
J.~C. Willems, ``Dissipative dynamical systems part {I}: General theory,''
  \emph{Arch. Ration. Mech. Anal.}, vol.~45, no.~5, pp. 321--351, 1972.

\bibitem{hill1980dissipative}
D.~J. Hill and P.~J. Moylan, ``Dissipative dynamical systems: Basic
  input-output and state properties,'' \emph{J. Franklin Inst.}, vol. 309,
  no.~5, pp. 327--357, 1980.

\bibitem{brogliato2020dissipative}
B.~Brogliato, R.~Lozano, B.~Maschke, and O.~Egeland, \emph{Dissipative Systems
  Analysis and Control: Theory and Applications}.\hskip 1em plus 0.5em minus
  0.4em\relax Springer, 2020.

\bibitem{khalil2015nonlinear}
H.~K. Khalil, \emph{Nonlinear Control}.\hskip 1em plus 0.5em minus 0.4em\relax
  Pearson Education, 2015.

\bibitem{moreschini2025moment}
A.~Moreschini, M.~Scandella, A.~Astolfi, and T.~Parisini, ``Moment matching by
  kernel-based learning,'' \emph{IEEE Trans. Autom. Control}, vol.~71, no.~4,
  pp. 2123--2138, 2026.

\bibitem{chen2021small}
W.~Chen, D.~Wang, and L.~Qiu, ``Small phase theorem,'' in \emph{Encycl. Syst.
  Control}, 2021, pp. 2082--2086.

\bibitem{moreschini2024generalized}
A.~Moreschini, M.~Bin, A.~Astolfi, and T.~Parisini, ``A generalized passivity
  theory over abstract time domains,'' \emph{IEEE Trans. Autom. Control},
  vol.~70, no.~1, pp. 2--17, 2025.

\bibitem{padoan2017eigenvalues}
A.~Padoan and A.~Astolfi, ``Eigenvalues and poles of nonlinear systems: A
  geometric approach,'' in \emph{Proc. 56th IEEE Conf. Decis. Control (CDC)},
  2017, pp. 2575--2580.

\bibitem{moreschini2025closed}
A.~Moreschini and A.~Astolfi, ``Closed-loop interpolation by moment matching
  for linear and nonlinear systems,'' \emph{IEEE Trans. Autom. Control},
  vol.~70, no.~5, pp. 2918--2933, 2025.

\bibitem{scarciotti2024survey}
G.~Scarciotti and A.~Astolfi, ``Interconnection-based model order reduction - a
  survey,'' \emph{Eur. J. Control}, vol.~75, p. 100929, 2024.

\bibitem{kuvcera1974matrix}
V.~Ku{\v{c}}era, ``The matrix equation {$AX+XB=C$},'' \emph{SIAM J. Appl.
  Math.}, vol.~26, no.~1, pp. 15--25, 1974.

\bibitem{de1981controllability}
E.~de~Souza and S.~Bhattacharyya, ``Controllability, observability and the
  solution of {$AX-XB= C$},'' \emph{Linear Algebra Appl.}, vol.~39, pp.
  167--188, 1981.

\bibitem{hu2006polynomial}
Q.~Hu and D.~Cheng, ``The polynomial solution to the {Sylvester} matrix
  equation,'' \emph{Appl. Math. Lett.}, vol.~19, no.~9, pp. 859--864, 2006.

\bibitem{Gelfand1964}
I.~M. Gel'fand and G.~E. Shilov, \emph{Generalized Functions, Volume 1:
  Properties and Operations}.\hskip 1em plus 0.5em minus 0.4em\relax AMS
  Chelsea Publishing, 1964.

\end{thebibliography}

\appendix
\label{sec:appendix}

Uniqueness of the solution of the Sylvester equation~\eqref{eq:Sylvester} is ensured by the fact that $A$ and $S$ do not share eigenvalues~\cite{kuvcera1974matrix}.
Controllability of the pair $(A,B)$ and observability of the pair $(S,L)$ imply that the solution is full rank~\cite{de1981controllability}. 

The following result specializes the closed-form solution to the Sylvester given in~\cite{hu2006polynomial}.
\begin{lemma} \label{th:sylvester_sol}
	Let $\sigma(A) \cap \sigma(S) = \varnothing$.
	For each $\varpi \in \R^+$, the unique solution of the Sylvester equation~\eqref{eq:Sylvester} is given by
	\begin{align} \label{eq:sol_sylvester}
		\Phi(\varpi) = - (A^2 + \varpi^2 I_n)^{-1} (BLS(\varpi) + ABL).
	\end{align}
\end{lemma}
\begin{proof} \label{th:sylvester_sol::proof}
	The Sylvester equation~\eqref{eq:Sylvester} has a unique solution by assumption that $\sigma(A) \cap \sigma(S) = \varnothing$,~\cite{kuvcera1974matrix}.
	This condition is satisfied by the assumption that $\varpi \in \R^+$, hence a unique solution $\Phi(\varpi)$ exists for all $\varpi \in \R^+$.
	By multiplying on the right the Sylvester equation~\eqref{eq:Sylvester} by $S(\varpi)$, and using the fact that $S(\varpi)^2 = -\varpi^2 I_{2m}$, we have that $A \Phi(\varpi) S(\varpi) = - BLS(\varpi) - \varpi^2 \Phi(\varpi) $.
	Similarly, we multiply on the left the Sylvester equation~\eqref{eq:Sylvester} by $A$ so that $A \Phi(\varpi) S(\varpi) = A^2 \Phi(\varpi) + A B L$.
	Hence, combining the two expressions through $A \Phi(\varpi) S(\varpi)$ we have
	\begin{align*}
		- BLS(\varpi) -\varpi^2 \Phi(\varpi) = A^2 \Phi(\varpi) + A B L .
	\end{align*}
	Finally, since by assumption $\varpi \in \R^+$, the matrix $(A^2 + \varpi^2 I_n)$ is invertible, we conclude~\eqref{eq:sol_sylvester}.
\end{proof}
\begin{lemma} \label{th:LTI-tf-Phi}
	Consider the LTI system~\eqref{eq:LTI-system} with transfer function $H : \C \to \C$, and let $\Phi(\varpi) = [\Phi_1(\varpi), \Phi_2(\varpi)] \in \R^{n \times 2}$ be the unique solution of the Sylvester equation~\eqref{eq:Sylvester} for $\varpi \in \R^+$.
	Then, for all $\varpi \in \R^+$, $H(j\varpi) = C\Phi_1(\varpi) + j C\Phi_2(\varpi)$.
\end{lemma}
\begin{proof} \label{th:LTI-tf-Phi::proof}
	Let $\iota \coloneq [1, j]^\top$.
	Using $S(\varpi) \iota = j \varpi \iota$ and $L \varpi = 1$, multiplying~\eqref{eq:sol_sylvester} on the right by $\iota$ gives:
	\begin{align*}
		\Phi(\varpi) \iota
		&
		=
		- (A^2 + \varpi^2 I_n)^{-1} (BLS(\varpi) + ABL) \iota
		\\
		&
		= \big((j\varpi I_n + A)(j\varpi I_n - A)\big)^{-1} (j\varpi I_n + A)B
		\\
		&
		= (j\varpi I_n - A)^{-1}B.
	\end{align*}
	Finally, since $\Phi(\varpi) \iota = \Phi_1(\varpi) + j\Phi_2(\varpi)$ and $C\Phi(\varpi) \iota = H(j\varpi)$, the claim follows.
\end{proof}
\begin{lemma} \label{th:r1_convolution}
	Consider the $\omega$-radius $r : \mathcal{N} \to \R$ of the system~\eqref{eq:driven_sys} to the input~\eqref{eq:driving_sys}, as defined in~\ref{def:w_radius}, for an $\omega \in \mathcal{N}$.
	Then, if $r(\omega) = 1$, there exists a generalized function $h_\omega : \R_+ \to \R$ such that
	\begin{equation}\label{eq:Y_r=1}
		Y(\omega, z)
		=
		\int_{\T_\omega}
		h_\omega(t - \tau) \ell(\omega, z(\tau)) \,d \tau.
	\end{equation}
\end{lemma}
\begin{proof}
	For brevity, let us recall that $u = \ell(\omega,z)$ and $\dot{u} = \dot{\ell}(\omega,z)$, and that
	$u(t) = \int_{\T_\omega} \delta(t - \tau) u(\tau) \, d\tau$,
	where $\delta$ is the Dirac delta function.
	Then, the time-derivative $\dot{\delta}$ of $\delta$ (for a formal definition of $\dot{\delta}$, see~\cite{Gelfand1964}), enjoys the property
	\begin{align*}
		&
		\int_{\T_\omega} \dot{\delta}(t - \tau) u(\tau) \, d\tau
		\\[-0.25em]
		&
		\hspace{6ex}
		=
		\delta(t - T_\omega) u(T_\omega) - \delta(t) u(0)
		-
		\int_{\T_\omega} \delta(t - \tau) \dot{u}(\tau) \, d\tau
		\\[-0.25em]
		&
		\hspace{6ex}
		=
		-\dot{u}(t).
	\end{align*}
	Then, as shown in Theorem~\ref{th:w-radius}, we have
	\begin{align*}
		Y(\omega, z(t))
		&
		=
		\overline{Y}(\omega, z(t))
		=
		\lambda_1 u(t)
		+
		\lambda_2 \dot{u}(t)
		\\
		&
		\hspace{-4ex}
		=
		\lambda_1 \int_{\T_\omega} \delta(t - \tau) u(\tau) \, d\tau
		-
		\lambda_2 \int_{\T_\omega} \dot{\delta}(t - \tau) u(\tau) \, d\tau
		\\
		&
		\hspace{-4ex}
		=
		\int_{\T_\omega} \pr*{ \lambda_1 \delta(t - \tau) + \lambda_2 \dot{\delta}(t - \tau)} u(\tau) \, d\tau.
	\end{align*}
	Finally,~\eqref{eq:Y_r=1} holds with
	$h_\omega(t) = \lambda_1 \delta(t) + \lambda_2 \dot{\delta}(t)$.
\end{proof}

\end{document}